\documentclass[aps,%
twocolumn,%
floatfix,%
pra,%
amsfonts,
groupedaddress,
superscriptaddress]{revtex4-2}%
\usepackage[utf8]{inputenc}
\usepackage[T1]{fontenc}
\usepackage{mathrsfs}       
\usepackage{mathtools}
\usepackage{amsmath}
\usepackage{amsfonts}       
\usepackage[titletoc,title]{appendix}
\usepackage{algorithm,algorithmic}

\usepackage{bbm}            
\usepackage{bm}             
\usepackage{comment}
\usepackage{amsthm}         
\usepackage{makecell}
\usepackage[dvipsnames]{xcolor}
\definecolor{beamer@blendedblue}{rgb}{0.2,0.2,0.7}
\usepackage{times}
\usepackage{caption}
\usepackage{booktabs} 

\usepackage[colorlinks=true,%
      linkcolor=beamer@blendedblue,%
      bookmarks=true,%
      breaklinks=true,%
      filecolor=beamer@blendedblue,%
      anchorcolor=yellow,%
      citecolor=beamer@blendedblue,%
      urlcolor=beamer@blendedblue,%
      pdfauthor={Shuanghong Tang, Congcong Zheng, and Kun Wang},%
      pdfsubject={Detecting and Eliminating the Quantum Noise of Quantum Measurements},%
      CJKbookmarks=true]{hyperref}

\usepackage{float}    
\newtheorem{definition}{Definition}
\newtheorem{proposition}[definition]{Proposition}

\newtheorem{lemma}[definition]{Lemma}
\newtheorem{theorem}[definition]{Theorem}

\mathchardef\ordinarycolon\mathcode`\:
\mathcode`\:=\string"8000
\def\vcentcolon{\mathrel{\mathop\ordinarycolon}}
\begingroup \catcode`\:=\active
  \lowercase{\endgroup
  \let :\vcentcolon
  }

\DeclareFontFamily{U}{mathx}{\hyphenchar\font45}
\DeclareFontShape{U}{mathx}{m}{n}{<-> mathx10}{}
\DeclareSymbolFont{mathx}{U}{mathx}{m}{n}
\DeclareMathAccent{\widebar}{0}{mathx}{"73}

\newcommand{\wt}[1]{\widetilde{#1}}
\newcommand{\wh}[1]{\widehat{#1}}

\newcommand{\ket}[1]{\left\vert{#1}\right\rangle}
\newcommand{\bra}[1]{\left\langle{#1}\right\vert}
\newcommand{\ketbra}[2]{\vert{#1}\rangle\!\langle{#2}\vert}
\newcommand\proj[1]{\vert{#1}\rangle\!\langle{#1}\vert}
\newcommand{\linear}[1]{\mathscr{L}(#1)}

\newcommand{\pos}[1]{\mathscr{P}(#1)}
\newcommand{\density}[1]{\mathscr{D}(#1)}
\newcommand{\opn}[1]{\operatorname{#1}}

\DeclareMathOperator{\tr}{Tr}  
\newcommand{\1}{\mathbbm{1}}

\makeatletter
\newsavebox{\@brx}
\newcommand{\llangle}[1][]{\savebox{\@brx}{\(\m@th{#1\langle}\)}%
  \mathopen{\copy\@brx\kern-0.5\wd\@brx\usebox{\@brx}}}
\newcommand{\rrangle}[1][]{\savebox{\@brx}{\(\m@th{#1\rangle}\)}%
  \mathclose{\copy\@brx\kern-0.5\wd\@brx\usebox{\@brx}}}
\makeatother

\newcommand*{\cH}{\mathcal{H}}

\newcommand*{\cL}{\mathcal{L}}
\newcommand*{\cM}{\mathcal{M}}

\newcommand*{\cS}{\mathcal{S}}
\newcommand*{\cT}{\mathcal{T}}
\newcommand*{\cU}{\mathcal{U}}

\begin{document}

\title{{Detecting and Eliminating Quantum Noise of Quantum Measurements}}

\author{Shuanghong Tang}
\affiliation{Institute for Quantum Computing, Baidu Research, Beijing 100193, China}

\author{Congcong Zheng}
\affiliation{Institute for Quantum Computing, Baidu Research, Beijing 100193, China}

\author{Kun Wang}
\email{wangkun28@baidu.com}
\affiliation{Institute for Quantum Computing, Baidu Research, Beijing 100193, China}

\begin{abstract}
Quantum measurements are crucial for extracting information from quantum systems, 
but they are error-prone due to hardware imperfections in near-term devices. 
Measurement errors can be mitigated through classical post-processing, based on the assumption of a classical noise model. 
However, the coherence of quantum measurements leads to unavoidable quantum noise that defies this assumption.
In this work, we introduce a two-stage procedure to systematically tackle such quantum noise in measurements. 
The idea is intuitive: we first detect and then eliminate quantum noise. 
In the first stage, inspired by coherence witness in the resource theory of quantum coherence, we design an efficient method to detect quantum noise.
It works by fitting the difference between two measurement statistics to the Fourier series,
where the statistics are obtained using maximally coherent states with relative phase
and maximally mixed states as inputs. 
The fitting coefficients quantitatively benchmark quantum noise.
In the second stage, we design various methods to eliminate quantum noise,
inspired by the Pauli twirling technique.
They work by executing randomly sampled Pauli gates before the measurement device and conditionally
flipping the measurement outcomes in such a way that the \textit{effective} measurement device contains only classical noise.
We numerically demonstrate the two-stage procedure's feasibility on the Baidu Quantum Platform. 
Notably, the results reveal significant suppression of quantum noise in measurement devices and substantial enhancement in quantum computation accuracy.
We highlight that the two-stage procedure complements
existing measurement error mitigation techniques,
and they together form a standard toolbox for manipulating measurement errors in near-term quantum devices. 
\end{abstract}

\maketitle

\section{Introduction}
Quantum computers offer significant potential across scientific and industrial domains. 
However, the current noisy intermediate-scale quantum (NISQ) computers~\cite{preskill2018quantum} introduce notable errors, 
necessitating their mitigation prior to engaging in practically valuable endeavors.
These errors emerge from either undesired qubit-environment interactions or 
the physical imperfections inherent in qubit initializations, quantum gates, and measurements~\cite{kandala2017hardware,arute2019quantum,google2020hartree,chen2021exponential}. 
Typically, errors in a quantum computer are categorized into quantum gate errors and measurement errors.
For the former, various quantum error mitigation techniques
have been proposed to mitigate the damages caused by errors
on near-term quantum devices
~\cite{temme2017error,endo2018practical,li2017efficient,mcclean2017hybrid,McClean2020Decoding,mcardle2019error,bonet2018low,he2020resource,giurgica2020digital,kandala2019error,endo2021hybrid,sun2021mitigating,czarnik2021error,takagi2021optimal,jiang2021physical,wang2023mitigating}.
For the latter, most experiments work with the assumption that measurement errors in quantum devices
are well understood in terms of \emph{classical noise}
models~\cite{chow2012universal,geller2020rigorous,geller2021conditionally}.
Specifically, an $n$-qubit noisy measurement device
can be characterized by a noise matrix of size $2^n\times 2^n$.
The element in the $\bm{x}$-th row and $\bm{y}$-th column of the matrix is
the probability of obtaining an outcome $\bm{x}$ provided that the true outcome is $\bm{y}$.
If one has access to this stochastic matrix, it is straightforward to classically reverse the noise effects
simply by multiplying the probability vector obtained from measurement statistics
by this matrix's inversion or its approximations, known as the
measurement error mitigation~\cite{chen2019detector,tannu2019mitigating,nachman2020unfolding,maciejewski2020mitigation,hicks2021readout,bravyi2021mitigating,murali2020software,kwon2020hybrid,funcke2022measurement,zheng2023bayesian,maciejewski2021modeling,barron2020measurement,wang2023mitigating}.

From the perspective of positive operator-valued measure (POVM) formalism,
the classical noise model indicates that 
the POVM elements characterizing the measurement device possess only non-zero diagonal values with respect to (w.r.t.) the computational basis. 
However, due to the coherent nature of quantum mechanics incurred during calibration and/or interaction, 
the POVM elements signalizing a near-term measurement device necessarily have off-diagonal non-zero values, 
which is called coherence in the resource theory of quantum coherence~\cite{streltsov2017colloquium}. 
That is to say, the coherence effect can be seen as one of the sources of measurement imperfections. 
This coherence effect renders the behavior of measurement devices unpredictable. 
Complete information regarding these devices can only be obtained through exponentially resource-intensive quantum detector tomography~\cite{fiuravsek2001maximum,lundeen2009tomography}.
Furthermore, this observation prompts the realization that existing measurement error mitigation techniques could be enhanced, 
as the classical noise model assumption fails to consider the influence of quantum noise. 
Consequently, we describe a measurement device as afflicted by \emph{quantum noise} when its POVM elements encompass non-zero off-diagonal values. 
This leads us to the fundamental query of how to effectively address this form of noise. 

\begin{figure*}[!htbp]
	\centering
	\includegraphics[width=0.8\textwidth]{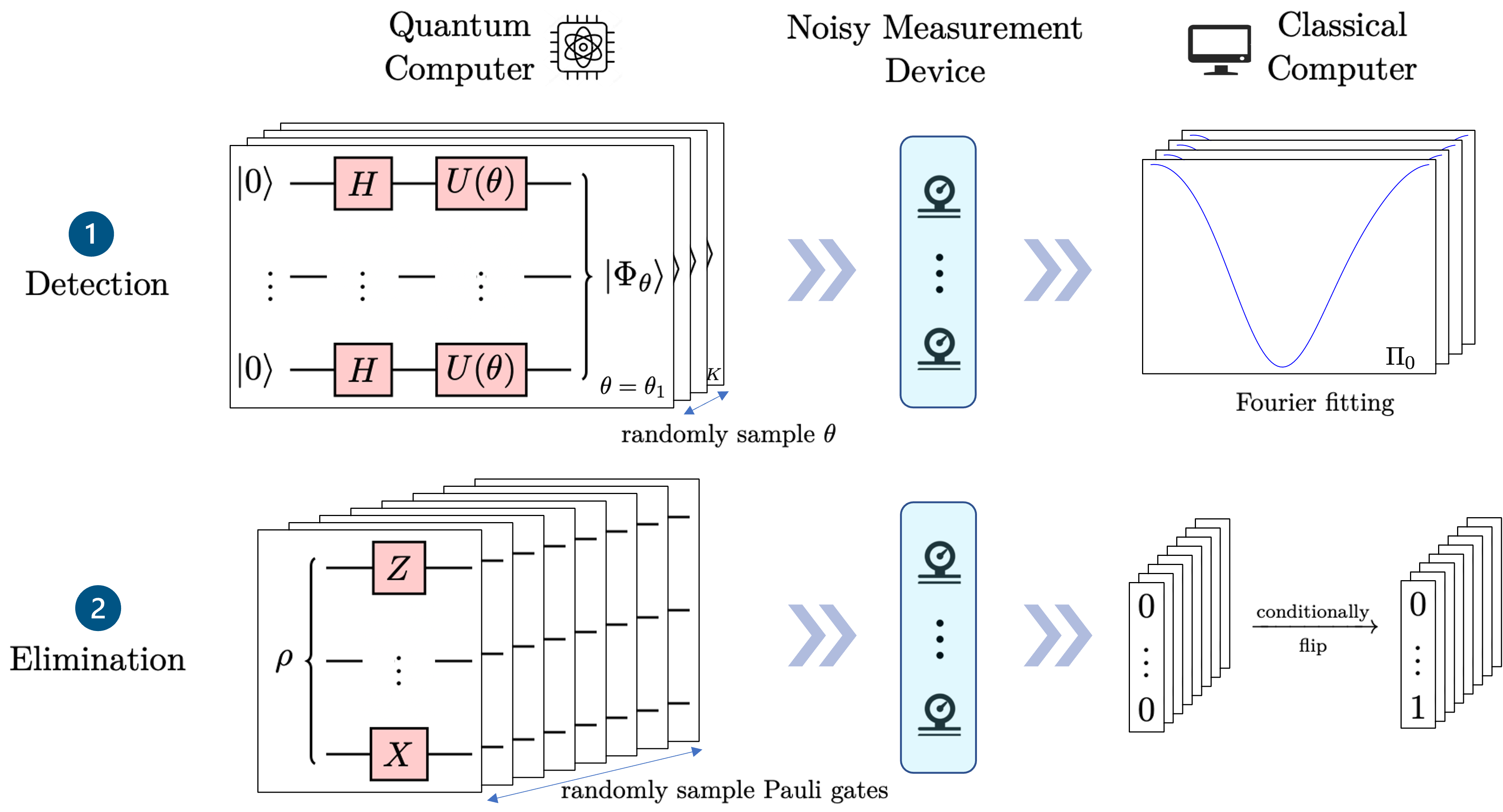}
	\caption{\raggedright 
			A two-stage procedure for detecting and eliminating quantum noise of quantum measurement devices.
			In the first stage, we prepare maximally coherent states $\vert\Phi_\theta\rangle$ with 
			relative phase $\theta$ and maximally mixed states as inputs, and fit the difference between 
			two measurement statistics to
			the Fourier series. The fitting coefficients quantitatively benchmark quantum noise.
			In the second stage, we execute randomly sampled Pauli gates before the measurement
			device and conditionally flip the outcomes,
			in such a way that the resulting effective measurement contains only classical noise.}
	\label{fig:2-stage}
\end{figure*}

In this work, we propose a two-stage procedure to systematically address quantum noise inherent in NISQ measurement devices. 
The procedure is illustrated in Figure~\ref{fig:2-stage} and
is very intuitive: 
we first detect and then eliminate quantum noise if there is any. 
It is noteworthy that the detection process is efficient, 
while the elimination process is resource-consuming. 
Consequently, prioritizing the detection of quantum noise is advisable. 
After the procedure, the classical noise model assumption is obviously satisfied
and the classical noise can be diminished using measurement error mitigation.
The rest of the paper is organized as follows.
Section~\ref{sec:Preliminaries} sets the notation and introduces
the task of addressing quantum noise inherent in quantum measurements.
Section~\ref{sec:efficient detection} elaborates the procedure's first stage, which proposes an efficient method to detect the quantum noise of measurement devices.
Section~\ref{sec:quantum noise elimination} elaborates the procedure's second stage
by describing three different methods to eliminate the quantum noise of measurement devices.
Section~\ref{sec:experimental-results} reports the feasibility of the proposed two-stage procedure
numerically on Baidu Quantum Platform~\cite{bqp2022} with three paradigmatic quantum applications.
The Appendices summarize technical details used in the main text.

\section{Preliminaries}\label{sec:Preliminaries}

In this section, we first set the notations.
Then, we review different yet equivalent mathematical formalisms of quantum measurements.
Finally, we rigorously define the classical and quantum noises of a quantum measurement device.

\subsection{Notations}

For a finite-dimensional Hilbert space $\cH$, we denote by $\linear{\cH}$ and $\pos{\cH}$ the linear and positive
semidefinite operators on $\cH$. Quantum states are in the set $\density{\cH}:=\{\rho\in\pos{\cH}\vert\tr\rho=1\}$.
For two operators $M, N\in\linear{\cH}$, we say $M\geq N$ if and only if $M-N\in\pos{\cH}$.
The identity matrix is denoted as $\1$, the maximally mixed state is denoted as $\Pi$, 
and the diagonal column vector of an operator $M$ w.r.t. the computational basis is denoted as $\opn{diag}(M)$.
Multipartite quantum systems are described by tensor product spaces.
Denote by $\mathbb{C}$, $\mathbb{R}$, and $\mathbb{R}^+$ the complex, real,
and non-negative real numbers, respectively.
In this paper, we assume $\cH_n$ represents the $n$-qubit Hilbert space
and our attention focuses on investigating classical and quantum noises of measurements in this space.

\subsection{Quantum measurements}

The most general kind of measurements in quantum mechanics is called positive operator-valued measures (POVMs).
A POVM is a set of operators $\{ E_{\bm{x}}\}_{\bm{x}\in\Sigma}$
satisfying $\forall \bm{x}, E_{\bm{x}}\geq0$ and $\sum_{\bm{x}} E_{\bm{x}}=\1$,
where $\Sigma$ is an alphabet and $\bm{x}\in\Sigma$ records the measurement outcome~\cite{nielsen2002quantum}.
If one performs the POVM $\{ E_{\bm{x}}\}$ on a quantum state $\rho$,
the probability of obtaining outcome $\bm{x}$ is given by the Born's rule
\begin{align}\label{eq:probability}
    p(\bm{x}) := \tr[ E_{\bm{x}}\rho].
\end{align}
Notice that in the POVM representation, we do not care about the post-measurement state
but only care about the probability of obtaining a particular outcome.
We fix $\Sigma=\{0,1\}^n$ in the following discussion.

In Quantum Information Theory, researchers commonly think of measurements from
the quantum channel viewpoint by regarding the post-measurement states as purely classical states.
Mathematically, the induced \emph{measurement channel} of a POVM $\{ E_{\bm{x}}\}_{\bm{x}\in\Sigma}$
is defined as~\cite{wilde2013quantum}
\begin{align}\label{eq:measurement-channel}
    \mathcal{M}(\rho) := \sum_{\bm{x}\in\Sigma}\tr\left[ E_{\bm{x}}\rho\right]\proj{\bm{x}},
\end{align}
where $\{\ket{\bm{x}}\}_{\bm{x}\in\Sigma}$ is the computational basis of the underlying Hilbert space.
Notice that a measurement channel takes a quantum system to a classical one.

The Pauli transfer matrix (PTM) representation is intensively used to notate
quantum states and quantum processes in tomography tasks~\cite{greenbaum2015introduction}.
In PTM, an $n$-qubit quantum channel is represented by a $4^n\times 4^n$ real matrix
defined w.r.t. vectorization in the Pauli basis instead of column-vectorization.
Specifically, the \emph{PTM matrix} $[\cM]$ of the measurement channel $\cM$
defined in~\eqref{eq:measurement-channel} has the form
\begin{align}\label{eq:ptm of measurement}
    [\cM]_{\bm{i}\bm{j}} := \frac{1}{2^n}\tr[P_{\bm{i}}\cM(P_{\bm{j}})]
    = \frac{1}{2^n}\sum_{\bm{x}\in\Sigma} \tr[ E_{\bm{x}}P_{\bm{j}}]\langle\bm{x}|P_{\bm{i}}|\bm{x}\rangle,
\end{align}
where $[\cM]_{\bm{i}\bm{j}}$ is the element in the $\bm{i}$-th row $\bm{j}$-th column and
$P_{\bm{i}}\in\mathsf{P}^n$ is the $\bm{i}$-th Pauli operator (sorted by lexicographic order)
in the $n$-qubit Pauli set $\mathsf{P}^n:=\{I,X,Y,Z\}^{\otimes n}$.

Given the POVM representation $\{ E_{\bm{x}}\}_{\bm{x}\in\Sigma}$,
we can compute the PTM representation via Eq.~\eqref{eq:ptm of measurement}.
In the following, we give a method to compute the POVM representation
given the PTM representation $[\cM]$. The proof is given in Appendix~\ref{appx:prop:from-PTM-to-POVM}.

\begin{proposition}\label{prop:from-PTM-to-POVM}
Let $[\cM]$ be the PTM representation of an arbitrary unknown POVM $\{ E_{\bm{x}}\}_{\bm{x}\in\Sigma}$.
It holds for arbitrary $\bm{x}\in\Sigma$ that
\begin{align}
     E_{\bm{x}} = \frac{1}{2^n}\sum_{P_{\bm{i}},P_{\bm{j}}\in\mathsf{P}^n}
    \bra{\bm{x}}P_{\bm{i}}\ket{\bm{x}}[\cM]_{\bm{i}\bm{j}}P_{\bm{j}}.
\end{align}
\end{proposition}

Throughout this paper we use the POVM representation $\{ E_{\bm{x}}\}_{\bm{x}\in\Sigma}$,
the measurement channel representation $\cM$, and the PTM representation $[\cM]$ interchangeably.
The representations of symbols should be clear from context.

\subsection{Classical and quantum noises}

In most existing quantum computing platforms~\cite{arute2019quantum,kandala2019error,postler2022demonstration},
measurement devices are designed to
implement an ideal $n$-qubit computational basis measurement, also called the $Z$ measurement,
whose POVM $\{ E_{\bm{x}}^i\}_{\bm{x}}$ has the form
\begin{align}\label{eq:ideal-measurement}
 \forall \bm{x}\in\{0,1\}^n,\quad  E_{\bm{x}}^i = \proj{\bm{x}},
\end{align}
where the superscript $i$ means \textit{ideal}
and $\{\ket{\bm{x}}\}_{\bm{x}\in\{0,1\}^n}$ is the computational basis of an $n$-qubit Hilbert space.
Note that $ E_{\bm{x}}^i$ is a $2^n\times 2^n$ matrix with $1$ in the $\bm{x}$-row and $\bm{x}$-column
and $0$ in all other cells.

Practically, the measurement apparatus is imperfect and introduces different kinds of noises.
In order to simplify the post-processing, experimenters commonly assume that
the incurred noises are well understood in terms of classical
noise models~\cite{chow2012universal,chen2019detector,geller2020rigorous,geller2021conditionally}.
In the POVM language, we say a POVM $\{ E_{\bm{x}}^c\}_{\bm{x}\in\{0,1\}^n}$ (the superscript $c$ means classical)
that characterizes a measurement device is \emph{classical} if,
besides $ E_{\bm{x}}^c\geq0$ and $\sum_{\bm{x}} E^c_{\bm{x}}=\1$, it satisfies two more conditions:
1) the off-diagonal elements of all $ E_{\bm{x}}^c$ are zero,
i.e., $\forall\bm{y}\neq\bm{z}$, $\bra{\bm{y}} E_{\bm{x}}\ket{\bm{z}} = 0$, $\forall\bm{x}$; and
2) there exists at least one $ E_{\bm{x}}^c$ whose diagonal elements contain more than one non-zero values,
i.e., $\exists\bm{x},\bm{y}$ such that $\bm{x}\neq\bm{y}$ and $\bra{\bm{y}} E_{\bm{x}}\ket{\bm{y}} \neq 0$.
Experimentally, this means that when we input the computational basis state $\ket{\bm{x}}$ to the measurement device, we have some probability of obtaining an incorrect measurement outcome $\bm{y}$.
If the measurement is classical, we call its induced measurement channel a \emph{classical measurement channel}.
We can use calibration to learn parameters of a classical measurement from experimental data
and error mitigation methods to diminish classical noise~\cite{maciejewski2020mitigation,tannu2019mitigating,nachman2020unfolding,hicks2021readout,bravyi2021mitigating,murali2020software,kwon2020hybrid,funcke2022measurement,zheng2023bayesian,maciejewski2021modeling,barron2020measurement,wang2023mitigating}.

In the most general case,
the POVM $\{ E_{\bm{x}}^q\}_{\bm{x}\in\{0,1\}^n}$ (the superscript $q$ means quantum)
that characterizes a noisy measurement device
only need to satisfy the fundamental conditions $ E_{\bm{x}}^q\geq0$ and $\sum_{\bm{x}} E^q_{\bm{x}}=\1$.
Thus, $ E_{\bm{x}}^q$ can have non-zero values in both diagonal and off-diagonal parts.
We term the non-zero values in the diagonal part as classical noise
and the non-zero values in the off-diagonal part as \emph{quantum noise}.
We can use quantum detector tomography to learn parameters of
a general POVM from experimental data~\cite{fiuravsek2001maximum,lundeen2009tomography},
which is both time consuming and computationally difficult.
The bitter truth is that in tomography, a few-qubit measurement device is already experimentally challenging.
What's worse, none of the mentioned mitigation methods can be applied to cancel the effect of quantum noise.
In the following, we propose a two-stage procedure to first detect and then eliminate quantum noise.
After the procedure, only classical noise remains, and we can use mitigation methods to handle them.

\section{Quantum noise detection}\label{sec:efficient detection}

In this section, we describe the first stage of the quantum noise manipulating procedure by proposing an efficient quantum noise detection method.
To commence, we introduce the notion of quantum noise witnesses,
which we define as quantum observables that comprehensively capture classical noise attributes and enable the physical identification of quantum noise. 
Subsequently, leveraging the direct measurability of quantum noise witnesses, 
we propose a Fourier series fitting method to detect quantum noise quantitatively. 
Ultimately, we validate the efficacy of this fitting methodology through numerical verification conducted on the Baidu Quantum Platform. 

\subsection{Quantum noise witness}

\subsubsection{Definition}

In brief, a quantum noise witness is a function designed to differentiate a particular quantum POVM element from classical POVM elements. 
These witnesses draw inspiration from quantum entanglement witnesses~\cite{horodecki2009quantum,guhne2009entanglement} and 
are rooted in geometry: the convex sets of classical POVM elements can be delineated using hyperplanes.

\begin{definition}[Quantum noise witness]\label{def:quantum noise-witness}
Let $W$ be a Hermitian operator in $\cH_n$. $W$ is called a quantum noise witness, if
\begin{enumerate}
 \item for arbitrary classical POVM $\{ E^c_{\bm{x}}\}_{\bm{x}}$,
 it holds for arbitrary $\bm{x}$ that $\tr[W E^c_{\bm{x}}] = 0$;
 \item there exists at least
 one quantum POVM $\{ E^q_{\bm{x}}\}_{\bm{x}}$
 such that there exists some $\bm{x}$ for which $\tr[W E^q_{\bm{x}}] \neq 0$.
\end{enumerate}
Thus, if one measures $\tr[W E_{\bm{x}}]\neq0$ for some $ E_{\bm{x}}$,
one knows for sure that this POVM element, and the corresponding POVM, is witnessed
by $W$ and contains quantum noise.
\end{definition}

Notice that our definition of quantum noise witness differs slightly from the standard definition used by entanglement witnesses. 
The separation hyperplane is determined by the expectation values equaling $0$. 
In Appendix~\ref{appx:quantum noise witness}, we elaborate that these two definitions are equivalent.

Motivated by the intuition that the operators close to a maximally coherent state
must possess non-zero off-diagonal values~\cite{streltsov2017colloquium},
one can construct a quantum noise witness from a given pure
coherent quantum state $\ket{\psi}:=\sum_{\bm{y}}c_{\bm{y}}\ket{\bm{y}}$,
where $c_{\bm{y}}\in\mathbb{C}$, via
\begin{align}\label{eq:psi-induced-witness}
 W_\psi := \alpha\1 - \proj{\psi},
\end{align}
where $\alpha\in\mathbb{R}^+$ is to be determined.
We call $\psi$ the \emph{probe state} and $W_\psi$ the \emph{$\psi$-induced quantum noise witness}.
The parameters $\{c_{\bm{y}}\}_{\bm{y}}$ and $\alpha$ must be chosen
to ensure that $W_\psi$ satisfies \textbf{Condition 1}. Notice that
\begin{align}\label{eq:JmfRrs}
 \tr\left[W_\psi E_{\bm{x}}\right]
= \sum_{\bm{y}}(\alpha- \vert c_{\bm{y}}\vert^2) E_{\bm{x}}(\bm{y},\bm{y})
- \sum_{\bm{y}\neq\bm{z}} c_{\bm{y}}^\ast c_{\bm{z}} E_{\bm{x}}(\bm{y},\bm{z}),
\end{align}
where $ A(\bm{y},\bm{z})$ is the element of $A$ in the $\bm{y}$-th row and $\bm{z}$-th column and
$c^\ast$ is the complex conjugate of $c$.
To ensure \textbf{Condition 1}, the first term in the RHS. of~\eqref{eq:JmfRrs}
must evaluate to $0$ for arbitrary classical $ E_{\bm{x}}$,
which holds if and only if $\vert c_{\bm{y}} \vert^2 = \alpha$ for arbitrary $\bm{y}$.
This leads to $\alpha = 1/2^n$ and $c_{\bm{y}} = e^{i\theta_{\bm{y}}}/\sqrt{2^n}$ for some $\theta_{\bm{y}}\in[0,2\pi]$,
thanks to the normalization condition. Correspondingly, the Hermitian operator
\begin{align}\label{eq:psi-induced-witness-2}
 W_\psi = \frac{1}{2^n}\1 - \proj{\psi}
\end{align}
witnesses the quantum noise of a given POVM via the expectation value
\begin{align}
 \tr\left[W_\psi E_{\bm{x}}\right]
&= - \frac{1}{2^n}
 \sum_{\bm{y}\neq\bm{z}} e^{i(\theta_{\bm{z}} - \theta_{\bm{y}})} E_{\bm{x}}(\bm{y},\bm{z}).\label{eq:difference}
\end{align}
Since $ E_x$ is positive semidefinite, the RHS. of~\eqref{eq:difference} must be real.
If $ E_{\bm{x}}$ is classical then $\tr\left[W_\psi E_{\bm{x}}\right]=0$.
Thus $\tr\left[W_\psi E_{\bm{x}}\right]\neq0$ implies that $ E_{\bm{x}}$ possesses quantum noise. It is not hard to find out that the quantum noise witness scheme can be easily realized through constructing special quantum states.

Since the expectation value of an observable depends linearly on the operator, the set of POVM elements is cut into two parts by the expectation value $\tr[W_\psi E]$.
In part with $\tr[W_\psi E]=0$ lies the set of all classical POVM elements and some quantum POVM elements that can not be detected by $W_\psi$, while the other part $\tr[W_\psi E]\neq0$ lies the set of quantum POVM elements detected by $W_\psi$.
From this geometrical interpretation,
it follows that all quantum POVMs can be detected by quantum noise witnesses
of the form $W_\psi$~\eqref{eq:psi-induced-witness-2}.
The proof is given in Appendix~\ref{appx:completeness}.

\begin{proposition}[Completeness]\label{prop:completeness}
For arbitrary quantum POVM element $ E^q$,
there exists a probe state $\psi$ for
which the $\psi$-induced quantum noise witness $W_\psi$~\eqref{eq:psi-induced-witness-2} detects $ E^q$.
\end{proposition}

Furthermore, as evident from Eq.~\eqref{eq:difference},
the $\psi$-induced quantum noise witness $W_\psi$ can not only detect but also
quantitatively measure the quantum noise of the noisy measurements to some degree.
Inspired by this observation, we propose the following quantum noise measure.

\begin{definition}[$\psi$-induced quantum noise measure]\label{def:psi induced quantum noise measure}
Let $\bm{ E}\equiv\{ E_{\bm{x}}\}_{\bm{x}\in\{0,1\}^n}$ be a POVM in $\cH_n$.
The $\psi$-induced quantum noise measure of $ E_{\bm{x}}$ is defined as
\begin{align}\label{eq:psi-quantum noise measure}
 \mathscr{Q}_\psi( E_{\bm{x}})
:= 2^n\vert\tr[W_\psi E_{\bm{x}}]\vert.
\end{align}
The (average) quantum noise measure of $\bm{ E}$ is defined as
\begin{align}\label{eq:psi-quantum noise measure 2}
 \mathscr{Q}_\psi(\bm{ E})
:= \frac{1}{2^n}\sum_{\bm{x}}\mathscr{Q}_\psi( E_{\bm{x}}).
\end{align}
\end{definition}

In Appendix~\ref{sec:properties of quantum noise measure},
we elaborate on the $\psi$-induced quantum noise measure
$\mathscr{Q}_\psi$ in depth by investigating its general properties,
the analytical solution in the qubit case, and its relations to the
resource theory of quantum measurements~\cite{Baek_2020}.

\subsubsection{Concrete case}

Although Proposition~\ref{prop:completeness} ensures that any quantum noise in POVM can in principle be detected by some quantum noise witness, 
the challenge lies in constructing good witnesses that can detect as many quantum POVMs as possible.
In this section, we specialize a probe state
and introduce a single-parameter quantum noise witness family
that can quantitatively gauge the quantum noise strength of quantum POVMs.

We first define the following $n$-qubit maximally coherent quantum state
\begin{align}\label{eq:Phi-theta}
\ket{\Phi_\theta} := \ket{+_\theta}^{\otimes n},
\end{align}
where the single-qubit state $\ket{+_\theta}$ is defined as
\begin{align}
 \ket{+_\theta}
:= \frac{\ket{0}+e^{i\theta}\ket{1}}{\sqrt{2}} = \frac{1}{\sqrt{2}}\begin{bmatrix} 1 \\ e^{i\theta} \end{bmatrix}
\end{align}
and can be obtained by applying the Hadamard gate followed
by a phase shift gate $R_\theta:=\begin{bmatrix} 1 & 0 \\ 0 & e^{i\theta}\end{bmatrix}$.
Using $\Phi_\theta$ as a probe state, the induced quantum noise witnesses have the form
\begin{align}\label{eq:Phi-witness}
W_\Phi^\theta := \frac{1}{2^n}\1 - \proj{\Phi_\theta}.
\end{align}

Remarkably, we find that the $\Phi_\theta$-induced quantum noise measure
$\mathscr{Q}_\Phi^{\theta}( E_{\bm{x}}):=2^n\vert\tr[W_\Phi^\theta E_{\bm{x}}]\vert$
has an elegant expression in the form of the sine-cosine Fourier series, where the coefficients are completely determined by the real and imaginary parts of $ E_{\bm{x}}$'s off-diagonal elements.
We summarize this exciting finding in the following theorem
and give its proof in Appendix~\ref{appx:Phi-Psi-quantum noise}.

\begin{theorem}\label{thm:Phi-Psi-quantum noise}
Let $\bm{ E}=\{ E_{\bm{x}}\}_{\bm{x}\in\{0,1\}^n}$ be a POVM in $\cH_n$.
For arbitrary $\bm{x}$, it holds that
\begin{align}\label{eq:quantum noise witness}
 \mathscr{Q}_\Phi^{\theta}\left( E_{\bm{x}}\right)
&= 2\left\vert\sum_{\bm{y< z}}\bigg(
 -\Re\left[ E_{\bm{x}}(\bm{y},\bm{z})\right]\cos\left[h(\bm{y},\bm{z})\theta\right]\right.\nonumber \\
&\qquad\qquad\;\;\left.+ \Im\left[ E_{\bm{x}}(\bm{y},\bm{z})\right]\sin\left[h(\bm{y},\bm{z})\theta\right]\bigg)\right\vert,
\end{align}
where $\Re\left(x\right)$ and $\Im\left(x\right)$ represent the real and imaginary
part of a complex number $x$, respectively,
$h(\bm{y},\bm{z}):=\left\vert\bm{y}\right\vert-\left\vert\bm{z}\right\vert$,
and $\vert\bm{y}\vert$ represents the Hamming weight of $\bm{y}$.
\end{theorem}

The key idea behind the elegant form of $\mathscr{Q}_\Phi^{\theta}\left( E_{\bm{x}}\right)$
is that the probe state yields an orthogonal basis of the trigonometric functions:
\begin{align}
 \left\{ \cos[h(\bm{y},\bm{z})\theta], \sin[h(\bm{y},\bm{z})\theta]:
 h(\bm{y},\bm{z}) \in\left[0, \cdots, n\right]\right\},
\end{align}
where $n$ denotes the number of qubits.
The off-diagonal parts of each POVM element can be expanded in the basis,
and the coefficients are the sum over off-diagonal elements that have the same Hamming weight difference $h$.
The detection ability of $\mathscr{Q}_\Phi^{\theta}\left( E_{\bm{x}}\right)$ is completely determined by the basis.
For example, the $2$-qubit quantum noise witness $W_\Phi^\theta$ has the form
\begin{align}
 W_\Phi^\theta = \frac{1}{4}\left[
 \begin{array}{cccc}
 0 & -e^{-i\theta} & -e^{-i\theta} & -e^{-i2\theta} \\
 -e^{i\theta} & 0 & -1 & -e^{-i\theta} \\
 -e^{i\theta} & -1 & 0 & -e^{-i\theta} \\
 -e^{i2\theta} & -e^{i\theta} & -e^{i\theta} & 0
 \end{array}
 \right].
\end{align}
When $h=1$, the real and imaginary components of $\bra{01} E_{\bm{x}}\ket{00}$, $\bra{10} E_{\bm{x}}\ket{00}$, $\bra{11} E_{\bm{x}}\ket{01}$ and $\bra{11} E_{\bm{x}}\ket{10}$ as well as their conjugates would be expanded as the coefficients of $\cos{(\theta)}$ and $\sin{(\theta)}$ respectively.

Consider a special case of $\mathscr{Q}_\Phi^{\theta}$ by choosing $\theta=0$:
\begin{align}
 \mathscr{Q}_\Phi^{\theta=0}\left( E_{\bm{x}}\right)
&= \left\vert-2\sum_{\bm{y}<\bm{z}}
 \Re\left[ E_{\bm{x}}\left({\bm{y},\bm{z}}\right)\right]\right\vert.
 \label{eq:Phi-quantum noise}
\end{align}
Interestingly, it records the sum of the real parts of \emph{all} off-diagonal values of $ E_{\bm{x}}$.
If $\mathscr{Q}_\Phi^{\theta=0}\left( E_{\bm{x}}\right) \neq 0$, we can conclude safely that
$ E_{\bm{x}}$ possesses quantum noise, though the converse statement is not necessarily true.

\subsection{Efficient detection}\label{section:4}

The fact that $\psi$-induced quantum noise witnesses and measures are directly measurable quantities makes them appealing for detecting quantum noise in measurement devices from the experiment perspective.
Specifically, to implement the quantum noise witness $W_\Phi^\theta$, 
defined in Eq.~\eqref{eq:Phi-witness}, for some fixed $\theta$,
we suffice to prepare many copies of quantum states---the maximally mixed state $\Pi:=\1/2^n$
and the maximally coherent state $\Phi_\theta$---and feed them to
the noisy measurement device. The expectation value $\tr[W_\Phi^\theta E_{\bm{x}}]$ can
be easily estimated from the measurement statistics via classical postprocessing.
Furthermore, after we successfully estimate $\tr[W_\Phi^\theta E_{\bm{x}}]$
for many different values of $\theta$, we can fit these expectation values to a Fourier series.
It is clear from Eq.~\eqref{eq:quantum noise witness} that the obtained fitting coefficients quantitatively characterize
the real and imaginary parts of the POVM element under consideration.
The whole quantum noise detection procedure is outlined in Algorithm~\ref{alg:quantum noise detection}.

\begin{algorithm}[H]
\caption{Quantum noise detection}
\begin{algorithmic}[1] \label{alg:quantum noise detection}
\REQUIRE $\bm{ E}$: the $n$-qubit measurement device, \\
\hskip1.4em $K$: number of phase values $\theta$ to be sampled, \\
\hskip1.4em $N_{\rm shots}$: number of measurement shots.

\ENSURE Coefficients that quantify $\bm{ E}$'s quantum noise strength.

\FOR{$k = 1, \cdots, K$}

\STATE Randomly sample a phase $\theta_k\in[0, 2\pi]$;
\STATE Generate $N_{\rm shots}$ number of copies of the $n$-qubit maximally
coherent state $\Phi_\theta$ with relative phase $\theta_k$;
\STATE Perform the device $\bm{ E}$ on the above quantum states
 and record the number of occurrences $N_{\bm{x}}^{\theta_k}$ of outcome $\bm{x}$;
\STATE Generate $N_{\rm shots}$ number of copies of the $n$-qubit maximally mixed state $\Pi$;
\STATE Perform the device $\bm{ E}$ on the above quantum states
 and record the number of occurrences $M_{\bm{x}}$ of outcome $\bm{x}$;
\STATE For every $\bm{x}\in\{0,1\}^n$, estimate the expectation value
 \begin{align*}
 \hat{\mathscr{Q}}_{\Phi}^{\theta_k}( E_{\bm{x}}) = 2^n(M_{\bm{x}}-N_{\bm{x}}^{\theta_k})/N_{\rm shots},
 \end{align*}
 where the overhat indicates that it is an estimate.
\ENDFOR
\STATE For every $\bm{x}\in\{0,1\}^n$,
 fit the data $\{\hat{\mathscr{Q}}_{\Phi}^{\theta_k}( E_{\bm{x}}):k=1,\cdots,K\}$ to a Fourier series
 and obtain the fitting coefficients for $ E_{\bm{x}}$.
\STATE Output all the fitting coefficients.
\end{algorithmic}
\end{algorithm}

Now we analyze the sample complexity---the number of copies of quantum states
consumed---of Algorithm~\ref{alg:quantum noise detection}. 
Obviously, a total of $2KN_{\rm shots}$ number of copies of quantum states has to be prepared and measured. 
The critical question is, how large $N_{\rm shots}$ should be so that
the expectation value $\mathscr{Q}_{\Phi}^{\theta_k}( E_{\bm{x}})$ for given $\theta_k$ and $\bm{x}$ can be estimated to the desired precision $\varepsilon$?

Each measurement gives a probability $X_i$, then $\bar{X} = \frac{1}{N_{\rm shots}}(X_1 + \dots + X_{N_{\rm shots}})$, and $\operatorname{E}[\bar{X}] = \operatorname{Tr}[E_x\Phi_{\theta}^k]$. Each $X_i$ satisfies $X_i\in[0,1]$. Then from Hoeffding's inequality, we can have, for all $\epsilon\geq0$, 

\begin{align}
    \operatorname{Pr}(\vert \bar{X} - \operatorname{E}[\bar{X}] \vert\geq \epsilon ) \leq 2\exp{(-2N_{\rm shots}\epsilon^2)}.
\end{align}

Introduce failure probability $\delta$, we can obtain that

\begin{align}
    2\exp{(-2\epsilon^2/N_{\rm shots})}\leq \delta.
\end{align}

Solving the above equation giving us 

\begin{align}
  N_{\rm shots} \geq \frac{\log(2/\delta)}{2\epsilon^2}.
\end{align}

That is to say, we need at least $N_{\rm shots}=\log(2/\delta)/(2\varepsilon^2)$ samples 
to achieve precision $\varepsilon$ with confidence level $1-\delta$ for each $\theta_k$.
The number of $\theta$ values sampled also matters but we are not able to present a theoretical analysis. 
According to our experience, $100$ samples of $\theta$ are sufficient to reach good fit result.

We briefly remark on the quantum states preparation in Algorithm~\ref{alg:quantum noise detection}.
An $n$-qubit maximally coherent state $\Phi_\theta$ can be constructed
using only single-qubit Hadamard and phase shift gates,
which have high gate fidelity in NISQ devices.
We have two different methods to prepare the maximally mixed state $\Pi$:
1) we can simulate it by randomly sampling in the computational basis and averaging;
2) we can construct its $2n$-qubit purification and discard the ancilla $n$-qubits.
We expect advanced experimental methods that can prepare 
the maximally mixed states in a direct manner.

\subsection{Numerical simulation}\label{sec:detection-simulation}

Here we manifest the power of Algorithm~\ref{alg:quantum noise detection}
by carrying out numerical simulations on Baidu Quantum Platform~\cite{bqp2022}.
Since its ideal simulator does not contain errors,
we attach a rotation gate $R_y:=e^{-i\pi\sigma_y/40}$ above the $\hat{y}$ axis to each qubit before measurement,
as shown in Figure~\ref{fig:noisy-measurement}.
In this way, the $R_y$ gates together with an ideal measurement simulate a three-qubit noisy measurement
$\{ E_{\bm{x}}\}_{\bm{x}\in\{0,1\}^{3}}$ that contains both classical and quantum noise,
which we will call the \textit{Ry measurement} in the following.

\begin{figure}[!htbp]
 \centering
 \includegraphics[width=0.2\textwidth]{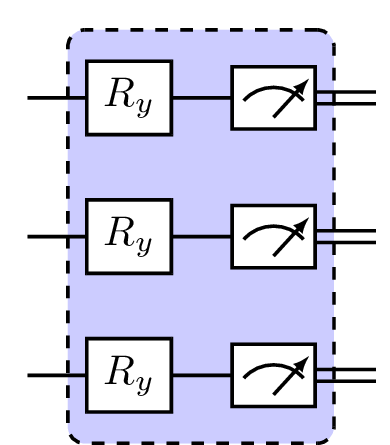}
 \caption{\raggedright A three-qubit Ry measurement is simulated by $R_y$ gates operate on qubits 
      followed by an ideal measurement.}
 \label{fig:noisy-measurement}
\end{figure}

For simplicity, we only apply Algorithm~\ref{alg:quantum noise detection} to
detect quantum noise of the first POVM element $ E_{000}$.
We remark that the same analysis is applicable to detect other POVM elements.
Since the Ry measurement is composed of $R_y$ gates and an ideal measurement,
its quantum noise measure can be analytically computed as
\begin{align}\label{eq:theoretical-value}
 \mathscr{Q}_\Phi^{\theta}\left( E_{000}\right) = 2\left\vert-0.018+0.236\cos(\theta)-0.018\cos(2\theta)\right\vert.
\end{align}
To collect experimental data, we uniformly sample $100$ $\theta$ values from $[0,2\pi]$.
For each $\theta$, we estimate $\hat{\mathscr{Q}}_{\Phi}^{\theta}( E_{000})$ with
a number of measurement shots $2^{13}$.
Additionally, we repeat the estimation procedure $10^3$ times to report the average value and the standard deviation.
Finally, we fit the estimated data to a Fourier series
and obtain the fitting coefficients for $ E_{000}$.
The analytical, experimental, and fitted data are visualized in Figure~\ref{fig:Ry-measurement}.
Notice that we have removed the absolute value restriction from the Figure for better illustration.
One can see from the figure that the fitted Fourier series (red line)
matches the theoretical curve (yellow line) pretty well,
though there are statistical errors in the experimental data.
That is to say, Algorithm~\ref{alg:quantum noise detection} is robust to statistical error and
reports the quantum noise measure $\mathscr{Q}_\Phi^{\theta}\left( E_{000}\right)$ with
high accuracy, validating its practicability in detecting quantum noise of measurement devices. 

\begin{figure}[!hbtp]
 \centering
 \includegraphics[width=0.5\textwidth]{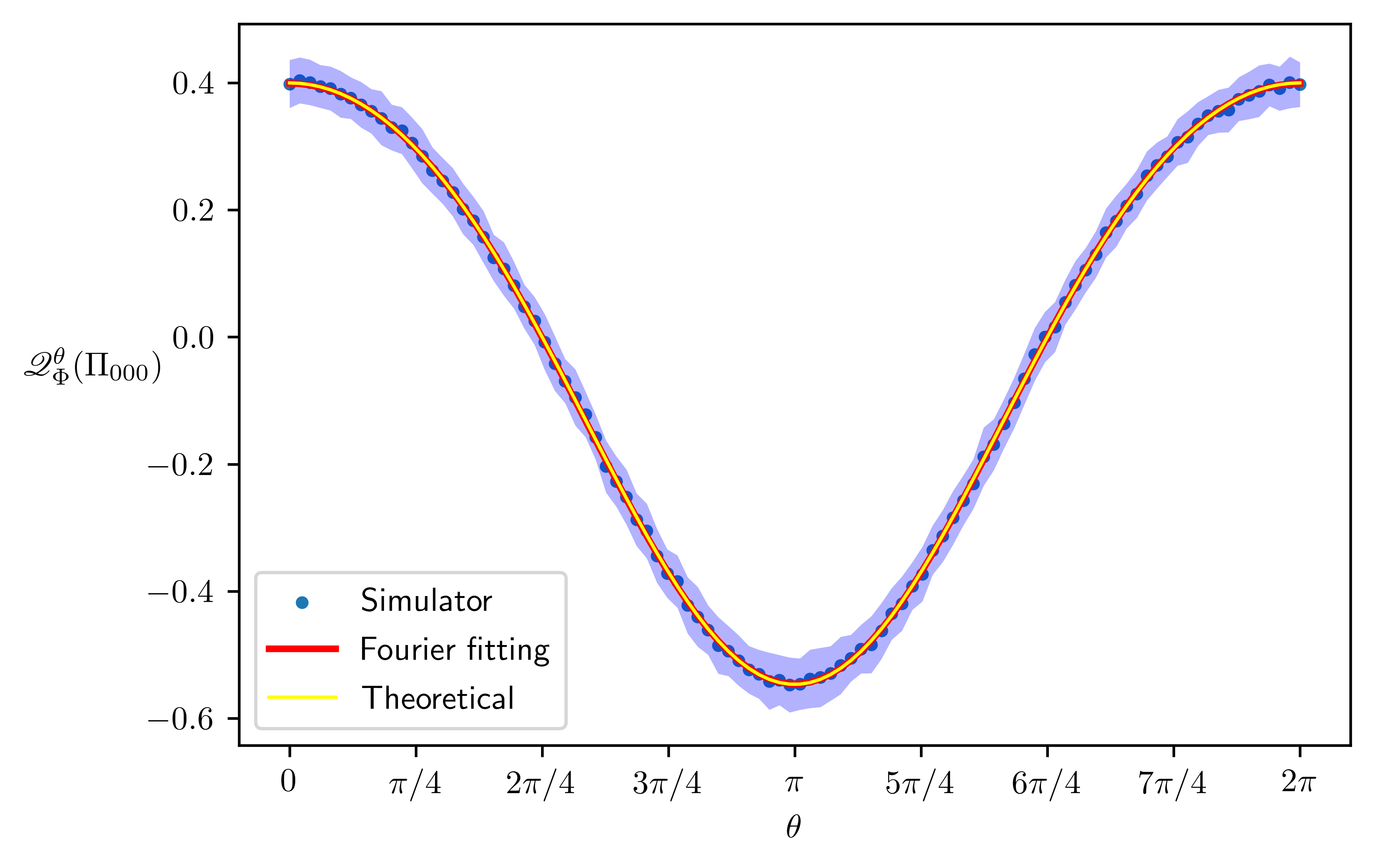}
 \caption{\raggedright Detecting quantum noise in the simulated Ry measurement.
 We uniformly sample $100$ $\theta$ from $[0, 2\pi]$.
 Experimental data on the \textbf{Simulator} are reported as the average value over $10^3$ independent estimations
 with standard error. \textbf{Fourier fitting} is plotted using the fitting coefficients obtained from
 the experimental data (cf. Algorithm~\ref{alg:quantum noise detection}).
 \textbf{Theoretical} is plotted using Eq.~\eqref{eq:theoretical-value}.}
 \label{fig:Ry-measurement}
\end{figure}

\section{Quantum noise elimination}\label{sec:quantum noise elimination}

This section describes the second stage of the quantum noise manipulation process, introducing three techniques for quantum noise elimination. 
The effectiveness of these methods is evaluated on the Baidu Quantum Platform.
They work by executing randomly sampled Pauli gates before the target measurement device
and conditionally flipping the measurement outcomes in such a way that the \emph{effective}
measurement---composed of random Pauli gates, the noisy measurement device, and
classical post-processing ---contains only classical noise.
These techniques are inspired by the established Pauli twirling approach~\cite{dur2005standard,wallman2016noise}. 
Notably, their fundamental distinction lies in the selection of the set of Pauli gates to be sampled.

Before presenting the technical details, we first briefly introduce the twirling technique.
For any set $\mathbb{T} \subset\cU(n)$, where $\cU(n)$ denotes the set of $n$-qubit unitary matrices,
we define the $\mathbb{T}$-twirl of a quantum channel $\cL$ as
\begin{align}\label{eq:twirling}
    \cL^{\mathbb{T}} := \frac{1}{\left|\mathbb{T}\right|} \sum_{T\in\mathbb{T}} \cT^\dagger\cL \cT,
\end{align}
where $\vert\cdot\vert$ is the size of the set and $\cT(\rho) := T \rho T^\dagger$.
Twirling involves the process of conjugating the noisy quantum channel $\mathcal{L}$ with 
a randomly selected gate $T$ from a predefined set of gates $\mathbb{T}$, 
termed the \emph{twirling set}, each time a quantum circuit is executed. 
This technique holds significant importance in the field of quantum information theory.
As an illustration, by choosing the twirling set to be the complete set of Pauli operators,
we can convert any quantum channel into a Pauli channel
whose noise elements correspond to the Pauli basis of the original noise~\cite{harper2020efficient,flammia2020efficient,harper2021fast,flammia2021pauli}.

\subsection{IZ dephasing}

Our first quantum noise elimination method is called the \emph{IZ dephasing}.
Interestingly, the idea of IZ dephasing comes from the observation that,
in an estimating expectation value task,
canceling the effect of the off-diagonal elements of the POVM elements is equivalent
to canceling the impact of the off-diagonal parts of input quantum states.
Now we formally elaborate on this observation.

The single-qubit completely dephasing channel $\Delta_1$ is defined as~\cite{nielsen2002quantum}
\begin{align}
    \Delta_1(\rho)
:= \frac{1}{2}\rho + \frac{1}{2}Z\rho Z^\dagger
 = \begin{bmatrix}
    \rho_{0,0} & 0 \\
    0 & \rho_{1,1}
   \end{bmatrix},
\end{align}
where $\rho_{i,j}$ is the element of $\rho$ in the $i$-th row and $j$-th column.
Intuitively, $\Delta_1$ erases all off-diagonal terms (w.r.t. the computational basis)
in $\rho$ while reserves the diagonal terms unchanged.
The $n$-qubit completely dephasing channel is defined to be
the $n$-th tensor product of $\Delta_1$:
\begin{align}\label{eq:completely dephasing channel}
    \Delta_n(\rho) := \Delta_1^{\otimes n}(\rho)
                  = \sum_{\bm{x}}\bra{\bm{x}}\rho\ket{\bm{x}}\proj{\bm{x}}.
\end{align}
Likewise, $\Delta_n$ erases all off-diagonal terms while reserving the diagonal terms unchanged.

For an arbitrary $n$-qubit quantum measurement channel $\cM$,
we define its \emph{$\Delta$-induced measurement channel} as
\begin{align}
    \cM^\Delta(\rho) := \sum_{\bm{x}}\tr[\Delta_n( E_{\bm{x}})\rho]\proj{\bm{x}}.
\end{align}
$\cM^\Delta$ is constructed from $\cM$ by erasing the off-diagonal elements of all POVM elements,
via the completely dephasing channel $\Delta_n$.
Notice that $\cM^\Delta$ is indeed a measurement since $\Delta_n( E_{\bm{x}})\geq0$
and
\begin{align}
  \sum_{\bm{x}}\Delta_n\left( E_{\bm{x}}\right)
  = \Delta_n\left(\sum_{\bm{x}} E_{\bm{x}}\right)
  = \Delta_n(\1)
  = \1,
\end{align}
where the last equality follows from the fact that $\Delta_n$ preserves the identity matrix.
What's more, $\cM^\Delta$ is classical since all of its POVM elements have only diagonal elements.
However, since neither $ E_{\bm{x}}$ nor $\Delta_n( E_{\bm{x}})$ are legitimate quantum states,
we are not able to realize the POVM elements $\Delta_n( E_{\bm{x}})$.
Using the fact that the adjoint map of $\Delta_n$ is itself, i.e., $\Delta_n^\dagger = \Delta_n$, we have
\begin{subequations}\label{eq:dephasing}
\begin{align}
    \cM^\Delta_n(\rho)
&= \sum_{\bm{x}}\tr[\Delta_n( E_{\bm{x}})\rho]\proj{\bm{x}} \\
&= \sum_{\bm{x}}\tr[ E_{\bm{x}}\Delta_n^\dagger(\rho)]\proj{\bm{x}} \\
&= \sum_{\bm{x}}\tr[ E_{\bm{x}}\Delta_n(\rho)]\proj{\bm{x}} \\
&= \cM(\Delta_n(\rho)).
\end{align}
\end{subequations}
This equation mathematically justifies our observation and indicates that
we can effectively implement the classical measurement $\cM^\Delta$ from the original measurement $\cM$
by completely dephasing the input quantum states first.
The above result is rigorously summarized in the following theorem.

\begin{theorem}[IZ dephasing]\label{thm: IZ dephasing}
Let $\cM$ be a measurement channel in $\cH_n$. The IZ-dephased linear map
\begin{align}\label{eq:IZ twirling measurement}
    \cM^{\operatorname{IZ}}(\cdot) := \cM^\Delta(\cdot) = \frac{1}{2^n}\sum_{P\in\{I,Z\}^{\otimes n}} \cM(P(\cdot) P),
\end{align}
is a classical measurement channel.
\end{theorem}

We show in Appendix~\ref{appx:iz-dephasing} that the IZ dephasing method
in Theorem~\ref{thm: IZ dephasing} can be equivalently understood within the twirling framework~\eqref{eq:twirling} via
\begin{align}\label{eq:IZ dephasing from twirling}
  \cM^{\operatorname{IZ}}(\cdot)
= \frac{1}{2^n}\sum_{P\in\{I,Z\}^{\otimes n}} P\cM(P(\cdot)P)P.
\end{align}
This alternative interpretation enables us to master all quantum noise elimination methods
in a unified way within the twirling framework, as we will show.

Operationally, Theorem~\ref{thm: IZ dephasing} motivates a simple experimental procedure
to eliminate the quantum noise of measurement devices.
Concretely, before sending the quantum states to the measurement device,
we randomly sample a Pauli gate $P$ from the set $\operatorname{IZ}\equiv\{I,Z\}^{\otimes n}$
and operate it on the quantum states. Theorem~\ref{thm: IZ dephasing} guarantees that
the averaged output classical state after measurement would not be contaminated by quantum
noise, though it still suffers from the classical noise of the measurement device.
Extra elimination costs are incurred due to random sampling and averaging.
The detailed procedure is summarized in Algorithm~\ref{alg:IZ dephasing}.

\begin{algorithm}[H]
\caption{Quantum noise elimination via IZ dephasing}
\begin{algorithmic}[1] \label{alg:IZ dephasing}
\REQUIRE $\bm{ E}$: the $n$-qubit measurement device, \\
\hskip1.4em $\rho$: the input $n$-qubit quantum state, \\
\hskip1.4em $K$: number of $n$-qubit Pauli operators to be sampled, \\
\hskip1.4em $N_{\rm shots}$: number of measurement shots.

\ENSURE $\hat{p}$, an estimate of $\opn{diag}(\rho)$ contaminated by classical noise.

\FOR{$k=1,\cdots,K$} 

\STATE Randomly sample a $n$-qubit Pauli gate $P_k\in\{I,Z\}^{\otimes n}$;

\STATE Generate $N_{\rm shots}$ number of copies of the quantum state $\rho$
      and apply the sampled Pauli gate $P_k$ to each of the quantum states;

\STATE Perform the device $\bm{ E}$ on the above quantum states
       and record the number of occurrences $N_{\bm{x}}^{(k)}$ of outcome $\bm{x}\in\{0,1\}^n$;
\STATE Estimate an empirical probability distribution $\hat{p}^{(k)}$ via
      \begin{align*}
          \hat{p}^{(k)}_{\bm{x}} = N_{\bm{x}}^{(k)}/N_{\rm shots},\; \forall\bm{x}\in\{0,1\}^n.
      \end{align*}
\ENDFOR

\STATE Calculate the mean value $\hat{p}$ from the estimated empirical probability
      distributions $\{\hat{p}^{(k)}:k=1,\cdots,K\}$ via
      \begin{align*}
          \hat{p}_{\bm{x}} = \frac{1}{K}\sum_{k=1}^K\hat{p}^{(k)}_{\bm{x}}, \; \forall\bm{x}\in\{0,1\}^n.
      \end{align*}
      Notice that $\hat{p}$ is an unbiased estimate of $\opn{diag}(\rho)$
        contaminated only by classical noise of $\bm{ E}$.
\STATE Output $\hat{p}$.
\end{algorithmic}
\end{algorithm}

\subsection{XY twirling}

Our second quantum noise elimination method is called the \emph{XY twirling}.
This method is inspired by a careful inspection on the PTM representation of measurement channels.
Specifically, we infer from Eq.~\eqref{eq:ptm of measurement} that,
the PTM $[\cM^c]$ of a classical measurement channel $\cM^c$ has a desirable property
that $[\cM^c]_{\bm{i}\bm{j}} = 0$ for arbitrary $P_{\bm{j}}\not\in\{I,Z\}^{\otimes n}$.
That is to say, only matrix elements whose column indices corresponding to Pauli operators
in the subset $\{I,Z\}^{\otimes n}$ have non-zero values.
On the other hand, the PTM $[\cM]$ of a general measurement channel $\cM$ does not satisfy such property.
There might exist some column index $\bm{j}$ such that
$P_{\bm{j}}\not\in\left\{I,Z\right\}^{\otimes n}$ and $\left[\cM\right]_{\bm{i}\bm{j}} \neq 0$.
This observation yields an intuitive way to deal with quantum noise:
If we can erase all those elements in the PTM matrix that are unique for quantum measurement,
we obtain a new PTM matrix that corresponds to a classical measurement.
Fortunately, we find that twirling an unknown measurement channel
with the set $\operatorname{XY} = \{X, Y\}^{\otimes n}$ can erase all those unique elements
and achieve our goal. We formally state this finding in the following theorem, and the proof is given in Appendix~\ref{appx:xy-twirling}.

\begin{theorem}[XY twirling]\label{thm: XY twirling}
Let $\cM$ be a measurement channel in $\cH_n$. The XY-twirled linear map
\begin{align}\label{eq:XY twirling measurement}
  \cM^{\operatorname{XY}}(\cdot)
:= \frac{1}{2^n}\sum_{P\in\{X,Y\}^{\otimes n}}  P\cM(P(\cdot) P)P
\end{align}
is a classical measurement channel.
\end{theorem}

\begin{algorithm}[H]
\caption{Quantum noise elimination via XY twirling}
\begin{algorithmic}[1] \label{alg:XY twirling}
\REQUIRE $\bm{ E}$: the $n$-qubit measurement device, \\
\hskip1.4em $\rho$: the input $n$-qubit quantum state, \\
\hskip1.4em $K$: number of $n$-qubit Pauli operators to be sampled, \\
\hskip1.4em $N_{\rm shots}$: number of measurement shots.

\ENSURE $\hat{p}$, an estimate of $\opn{diag}(\rho)$ contaminated by classical noise.

\FOR{$k=1,\cdots,K$} 

\STATE Randomly sample a $n$-qubit Pauli gate $P_k\in\{X,Y\}^{\otimes n}$;

\STATE Generate $N_{\rm shots}$ number of copies of the quantum state $\rho$
      and apply the sampled Pauli gate $P_k$ to each of the quantum states;

\STATE Perform the device $\bm{ E}$ on the above quantum states
       and record the number of occurrences $N_{\bm{x}}^{(k)}$ of outcome $\bm{x}\in\{0,1\}^n$;
\STATE Estimate an empirical probability distribution $\hat{p}^{(k)}$ via
      \begin{align*}
          \hat{p}^{(k)}_{\bm{x}} = N_{\overline{\bm{x}}}^{(k)}/N_{\rm shots},\; \forall\bm{x}\in\{0,1\}^n,
      \end{align*}
      where $\overline{\bm{x}}$ is obtained from $\bm{x}$ by flipping all bits.
\ENDFOR

\STATE Calculate the mean value $\hat{p}$ from the estimated empirical probability
      distributions $\{\hat{p}^{(k)}:k=1,\cdots,K\}$ via
      \begin{align*}
          \hat{p}_{\bm{x}} = \frac{1}{K}\sum_{k=1}^K\hat{p}^{(k)}_{\bm{x}}, \; \forall\bm{x}\in\{0,1\}^n.
      \end{align*}
      Notice that $\hat{p}$ is an unbiased estimate of $\opn{diag}(\rho)$
        contaminated only by classical noise of $\bm{ E}$.
\STATE Output $\hat{p}$.
\end{algorithmic}
\end{algorithm}

Much like Algorithm~\ref{alg:IZ dephasing},
Theorem~\ref{thm: XY twirling} stimulates the following experimental procedure
to eliminate the quantum noise of measurement devices.
Specifically, before sending the quantum states to the measurement device,
we randomly sample a Pauli gate $P$ from the twirling set $\operatorname{XY}\equiv\{X, Y\}^{\otimes n}$
and operate it on the quantum states. Then we measure the transformed states.
After measurement, we have to operate the Pauli gate $P$ again on the classical output states
to accomplish the twirling procedure,
which is different from Algorithm~\ref{alg:IZ dephasing}.
The tricky point is that since a measurement channel takes a quantum system to a classical one
and the logical output states are classical,
operating Pauli gates sampled from $\operatorname{XY}$
is equivalent to flipping the outcome bits, i.e., $0\to 1$ and $1\to 0$,
which can be done through classical post-processing.
Theorem~\ref{thm: XY twirling} guarantees that
the final averaged output classical state approximates
the diagonal part of the input quantum state $\rho$,
but is suffering from the classical noise of the measurement device.
The detailed procedure is summarized in Algorithm~\ref{alg:XY twirling}.

\subsection{Pauli twirling}

Our last quantum noise elimination method is the \emph{Pauli twirling}.
This method has the same spirit as the XY twirling method and aims to erase the
elements in the PTM matrix that are unique for quantum measurement.
Indeed, the idea of using Pauli twirling to eliminate quantum noise of measurement devices
has been previously conceived in~\cite{wallman2016noise}.
Here we present their proposal in a rigorous and experimental friendly way in the following proposition.
The proof is given in Appendix~\ref{appx:pauli-twirling}.

\begin{proposition}[Pauli twirling]\label{prop: Pauli twirling}
Let $\cM$ be a measurement channel in $\cH_n$. The Pauli-twirled linear map
\begin{align}\label{eq:Pauli twirling measurement}
  \cM^{\operatorname{Pauli}}(\cdot)
:= \frac{1}{4^n}\sum_{P\in\{I,X,Y,Z\}^{\otimes n}}  P\cM(P(\cdot) P)P
\end{align}
is a classical measurement channel.
\end{proposition}

Comparing Eqs.~\eqref{eq:IZ twirling measurement},~\eqref{eq:XY twirling measurement},
and~\eqref{eq:Pauli twirling measurement}, we can see that
the Pauli twirling method has a much larger twirling set that
incorporates the twirling sets of IZ dephasing and XY twirling.
As before, Proposition~\ref{prop: Pauli twirling} motivates a simple experimental procedure
to eliminate the quantum noise of measurement devices.
Concretely, before sending the quantum states to the measurement device,
we randomly sample a Pauli gate $P$ from the set $\operatorname{Pauli}\equiv\{I, X, Y, Z\}^{\otimes n}$
and operate it on the quantum states. Then we measure the transformed states.
After measurement,
we have to operate the Pauli gate $P$ again on the output quantum states to accomplish the twirling procedure.
Different from Algorithm~\ref{alg:XY twirling},
here we have to flip the outcome bits conditionally,
since $I$ and $Z$ gates have different behavior to $X$ and $Y$ gates on the computational basis states.
That is, we will flip the $n$-bit binary string $\bm{x}$ conditioned on the Pauli gate $P$:
if the $i$-th operator of $P$ is $I$ or $Z$, we do not flip the $i$-th bit of $\bm{x}$;
if the $i$-th operator of $P$ is $X$ or $Y$, we flip the $i$-th bit of $\bm{x}$.
Proposition~\ref{prop: Pauli twirling} guarantees that
the averaged output approximates the diagonal part of the input quantum state $\rho$,
which only suffers from the classical noise of the measurement device.
The detailed procedure is summarized in Algorithm~\ref{alg:Pauli twirling}.

\begin{algorithm}[H]
\caption{Quantum noise elimination via Pauli twirling}
\begin{algorithmic}[1] \label{alg:Pauli twirling}
\REQUIRE $\bm{ E}$: the $n$-qubit measurement device, \\
\hskip1.4em $\rho$: the input $n$-qubit quantum state, \\
\hskip1.4em $K$: number of $n$-qubit Pauli operators to be sampled, \\
\hskip1.4em $N_{\rm shots}$: number of measurement shots.

\ENSURE $\hat{p}$, an estimate of $\opn{diag}(\rho)$ contaminated by classical noise.

\FOR{$k=1,\cdots,K$} 

\STATE Randomly sample a $n$-qubit Pauli gate $P_k\in\{I,X,Y,Z\}^{\otimes n}$;

\STATE Generate $N_{\rm shots}$ number of copies of the quantum state $\rho$
      and apply the sampled Pauli gate $P_k$ to each of the quantum states;

\STATE Perform the device $\bm{ E}$ on the above quantum states
       and record the number of occurrences $N_{\bm{x}}^{(k)}$ of outcome $\bm{x}\in\{0,1\}^n$;
\STATE Estimate an empirical probability distribution $\hat{p}^{(k)}$ via
      \begin{align*}
          \hat{p}^{(k)}_{\bm{x}} = N_{\overline{\bm{x}}}^{(k)}/N_{\rm shots},\; \forall\bm{x}\in\{0,1\}^n,
      \end{align*}
      where $\overline{\bm{x}}$ is obtained from $\bm{x}$ by conditionally flipping the bits
      based on the Pauli operator $P_k$: If the $i$-th operator of $P_k$ is $X$ or $Y$,
      then flip the $i$-th bit of $\bm{x}$; otherwise, do not flip.
\ENDFOR

\STATE Calculate the mean value $\hat{p}$ from the estimated empirical probability
      distributions $\{\hat{p}^{(k)}:k=1,\cdots,K\}$ via
      \begin{align*}
          \hat{p}_{\bm{x}} = \frac{1}{K}\sum_{k=1}^K\hat{p}^{(k)}_{\bm{x}}, \; \forall\bm{x}\in\{0,1\}^n.
      \end{align*}
      Notice that $\hat{p}$ is an unbiased estimate of $\opn{diag}(\rho)$
        contaminated only by classical noise of $\bm{ E}$.
\STATE Output $\hat{p}$.
\end{algorithmic}
\end{algorithm}

\subsection{Comparisons}

We have proposed three methods to eliminate quantum noise in measurement devices.
Here we discuss the similarities, advantages, and differences among these methods.

First, we note that all these methods can be described in a unified way within the twirling framework.
Following the compiling idea of~\cite{wallman2016noise},
these methods can be experimentally implemented with only classical pre-and post-processing, without
\emph{truly} inserting the sampled Pauli gates before the measurement device.
More precisely, we can compile the sampled Pauli gate with the last gates in the quantum circuit
that generates the quantum state, optimizing two quantum gates into one.
This compiling requires only a tiny classical overhead in the compilation cost and
can be implemented on the fly with fast classical control.

The IZ dephasing method utilizes a twirling set of size $2^n$ and does not require bits to flip after measurement.
Besides, since we do not need to insert `I' gates in practice, the number of quantum gates
inserted is $n2^{n-1}$. What's more, from Eq.~\eqref{eq:dephasing} we can see that this method removes off-diagonal values but preserves diagonal values of the POVM elements.
That is, we do not change the classical behavior of the measurement device.
The XY twirling method utilizes a twirling set of size $2^n$ but requires bits to flip after measurement, which incurs extra classical possessing cost.
The total number of quantum gates inserted is $n2^n$, twice of the former method.
Moreover, it removes not only off-diagonal values but also changes diagonal values of the POVM elements,
causing unpredictable classical behavior.
The Pauli twirling method is the most expensive since
it utilizes a twirling set of size $4^n$ and requires a conditional flip after measurement.
The total number of quantum gates inserted is $n4^n$, much larger than the former methods.
It removes not only off-diagonal values but also regularizes diagonal values of
POVM elements, so that the diagonal values of different POVM elements are the same but reordered.

Notice that some methods not only erase the off-diagonal values but also change the diagonal values
of original POVM elements. We are interested in whether they will dramatically alter the
fidelity of the measurement device, which is essential since if these methods result in substantial
decay in the fidelity, we have to scarify fidelity for quantum noise.
Luckily, we can show that the measurement fidelity is robust to these methods
in the sense that it remains unchanged by these methods.
Before presenting the main result, let's first define the measurement fidelity.
The measurement fidelity, also known as readout fidelity and assignment fidelity,
of a quantum measurement $\cM$ (with respect to the computational basis measurement)
is defined as~\cite{magesan2015machine}
\begin{align}
    f(\cM) := \frac{1}{2^n}\sum_{\bm{x}\in\{0,1\}^n}\langle \bm{x}\vert E_{\bm{x}}\vert\bm{x}\rangle.
\end{align}
Intuitively, $f(\cM)$ quantifies how well $\cM$ preserves the computational basis states on average.
Then, we have the following identification among the measurement fidelities of the effective measurements.
The proof is given in Appendix~\ref{appx:fidelity-preserving}.

\begin{proposition}\label{prop:fidelity-preserving}
Let $\cM$ be a measurement channel in $\cH_n$. It holds that
\begin{align}
    f(\cM) = f(\cM^{\operatorname{IZ}}) = f(\cM^{\operatorname{XY}}) = f(\cM^{\operatorname{Pauli}}).
\end{align}
\end{proposition}

Proposition~\ref{prop:fidelity-preserving} implies that
although Pauli twirling employs many more Pauli gates to eliminate quantum noise,
it cannot make the resulting effective measurement channel in the sense of measurement fidelity.
One thus wonders if the Pauli twirling technique could bring any advantage.
Indeed, we show that Pauli twirling can simplify the classical noise process
by regularizing the POVM elements, as formally stated in the following proposition.
The proof is given in Appendix~\ref{appx:regularize povm}.

\begin{proposition}\label{prop:regularize povm}
Let $\cM$ be a measurement channel in $\cH_n$ and
$\cM^{\operatorname{Pauli}}$ be its Pauli-twirled channel.
It holds that
\begin{align}\label{eq:regularize povm}
    \operatorname{diag}( E_{\bm{y}}^{\operatorname{Pauli}}) = \mathsf{T}_{\bm{y}}\mathsf{T}_{\bm{x}}^{-1} \operatorname{diag}( E_{\bm{x}}^{\operatorname{Pauli}}),
\end{align}
where the transition matrix $\mathsf{T}_{{\bm{x}}}$ is defined as
\begin{align}
    \left(\mathsf{T}_{{\bm{x}}}\right)_{\bm{i}\bm{j}} := (-1)[({\bm{x}}+\bm{i})\cdot \bm{j}].
\end{align}
where $\bm{a}\cdot \bm{b}$ is the bitwise inner product between $\bm{a}$ and $\bm{b}$ and
\begin{align}
(-1)[\bm{a}] := \prod_{\bm{a}_i\in\bm{a}} (-1)^{\bm{a}_i}.
\end{align}
\end{proposition}

Proposition~\ref{prop:regularize povm} implies that there exist strong correlations
among the POVM elements of the Pauli-twirled measurement channel:
they share the same diagonal values, and their order is specified by the outcome label $\bm{x}$.
Thus, given the knowledge of an arbitrary POVM element,
we can completely infer the remaining POVM elements via Eq.~\eqref{eq:regularize povm}.
This lovely property can be employed in quantum detector tomography as a constraint
to improve the estimation accuracy of Pauli-twirled measurement channels.


\section{Experimental results}\label{sec:experimental-results}

In this section, we test the two-stage procedure with three quantum applications---estimating 
the expectation value of Mermin polynomial,
the fidelity of GHZ states, and the ground state energy of a hydrogen molecule---using
a noisy simulator with Ry measurements (cf. Figure~\ref{fig:Ry-measurement}) on Baidu Quantum Platform, 
to showcase its ability in improving computation accuracy.
The experimental data are collected via the QCompute Software Development Kit~\cite{bqp2022}.
All experiments were performed with $N_{\rm{shots}} = 2^{13}$ number of shots.

\subsection{Expectation value of Mermin polynomial}

\begin{table*}
\centering
\renewcommand{\arraystretch}{1.2} 
\setlength{\tabcolsep}{12pt} 
\begin{tabular}{@{}cccccc@{}}
        \toprule
     & \textbf{Unmitigated} & \textbf{Inverse} & \textbf{Least square} & \textbf{IBU} \\ \hline
    \textbf{Raw} &
      $10.775\pm 0.033$ &$11.335\pm0.035$ & $11.340\pm0.035$ &$11.335\pm0.035$ \\
    \textbf{IZ dephasing} &
      $10.768\pm0.008$ & $11.313\pm0.008$  &$11.318\pm0.008$ &$11.313\pm0.008$\\
    \textbf{XY twirling} &
      $10.766\pm0.008$ & $11.315\pm0.009$ & $11.320\pm0.009$ &$11.315\pm0.009$ \\
    \textbf{Pauli twirling} &
      $10.767\pm0.002$ & $11.315\pm0.002$ & $11.320\pm0.003$ &$11.315\pm0.002$ \\ \bottomrule
\end{tabular}
\caption{\raggedright
The four-qubit Mermin polynomial~\eqref{eq:Mermin polynomial} measured on Baidu Quantum Platform.
The theoretical maximum value allowed by quantum mechanics is $8\sqrt{2}\approx11.314$.
Experimental data are reported as the average value over $10^3$ independent estimations with standard error.
The \textbf{Raw}, \textbf{IZ dephasing}, \textbf{XY twirling}, and \textbf{Pauli twirling}
rows record the values without or with the corresponding quantum noise elimination methods.
The \textbf{Unmitigated}, \textbf{Inverse}, \textbf{Least square}, and \textbf{IBU} columns
record the values without or with the corresponding error mitigation methods.}
\label{tab:Mermin polynomial}
\end{table*}

Our first application is to estimate the expectation values of Mermin polynomials~\cite{mermin1990extreme},
which are one of the most significant examples to test non-local quantum correlations in multi-partite systems.
We notice that many research groups have measured Mermin polynomials on
superconducting quantum computers to assess their quantum reliability~\cite{neeley2010generation,dicarlo2010preparation,alsina2016experimental,garcia2018five,gonzalez2020revisiting,geller2021conditionally}.
We measure the following $4$-qubit Mermin polynomial
\begin{align}\label{eq:Mermin polynomial}
    M_4 &:= \left\langle XXXY\right\rangle + \left\langle XXYX\right\rangle + \left\langle XYXX\right\rangle + \left\langle YXXX\right\rangle\notag  \\
    &+ \left\langle XXYY\right\rangle + \left\langle XYXY\right\rangle + \left\langle XYYX\right\rangle + \left\langle YXXY\right\rangle\notag \\
    &+ \left\langle YXYX\right\rangle + \left\langle YYXX\right\rangle - \left\langle XXXX\right\rangle - \left\langle XYYY\right\rangle\notag  \\
    &- \left\langle YXYY\right\rangle - \left\langle YYXY\right\rangle - \left\langle YYYX\right\rangle + \left\langle YYYY\right\rangle
\end{align}
on the quantum state
\begin{align}\label{eq:Mermin-state}
    \ket{G} := \frac{\vert 0000\rangle+e^{3\pi i/4}\vert 1111\rangle}{\sqrt{2}}.
\end{align}
Note that the state $\ket{G}$ can be prepared by the quantum circuit
described in~\cite[Fig. 1(b)]{geller2021conditionally}.
We choose this state because it allows for a maximal violation of local realism.
Experimentally, the Mermin polynomial~\eqref{eq:Mermin polynomial} is measured
on the first four qubits of the simulator after preparing the quantum state~\eqref{eq:Mermin-state}.
We use the proposed elimination methods to cancel the effect of quantum noise.
Besides, we adopt four error mitigation 
methods---the inverse method (\textbf{Inverse})~\cite{maciejewski2020mitigation},
the least square method (\textbf{Least square})~\cite{chen2019detector}, and
the iterative Bayesian unfolding method (\textbf{IBU})~\cite{nachman2020unfolding}---to
reverse the classical noise effect and improve the estimation accuracy.
Note that we use the implementations of these mitigation methods in
Baidu Quantum Platform~\footnote{Specifically, we use the
\href{https://quantum-hub.baidu.com/qep/}{Quantum Error Processing}
toolkit developed on Baidu Quantum Platform.
It aims to deal with quantum errors inherent in quantum devices using software solutions and
offers various powerful quantum error processing tools.}.

The experimental results are summarized in Table~\ref{tab:Mermin polynomial}.
We find from the second column that the effects of both classical and quantum noises are significant,
as reflected in the differences between the raw and theoretical values.
Moreover, though the elimination methods alone cannot improve the estimation accuracy,
they considerably decrease the statistical error, making the estimation procedure more robust.
We find that the second row error mitigation methods in the presence of quantum noise 
would lead to an overestimation of the entanglement,
indicating that it should not be trusted in quantum foundations experiments like this.
We also find from the remaining rows that the difference in the estimated values
between the three elimination methods is statistically insignificant
and is insensitive to different error mitigation techniques.
In other words, all these methods can accurately eliminate quantum noise in the measurement device.
This justifies our expectation that the two-stage procedure and
error mitigation techniques form a standard toolbox for manipulating measurement errors.

\subsection{Fidelity of GHZ states}

Our second application is to estimate the fidelity of multi-partite GHZ states via parity oscillation,
which is a standard method for checking the entanglement of GHZ states~\cite{monz201114}.
The parity oscillation protocol works as follows:
first, we generate a $n$-qubit GHZ state using Hadamard and CNOT gates;
then, we apply the same rotation operation $U_{\phi}=e^{\mathrm{i}\pi\sigma_{\phi}/4}$ to all qubits
of the GHZ state, where $\sigma_{\phi} := \cos{\phi}X + \sin{\phi}Y$;
finally, we measure the qubits on the computational basis and
estimate the expectation value of observable $Z^{\otimes n}$.
By varying the phase $\phi$, we can observe an oscillation effect of the expectation values
and the oscillation intensity benchmarks the prepared GHZ state's entanglement quality.
Experimentally, we consider a $4$-qubit system, where the GHZ state has the form
\begin{align}\label{eq:GHZ}
    \ket{\opn{GHZ}} := \frac{\ket{0000}+\ket{1111}}{\sqrt{2}}.
\end{align}
We prepare this state on the first four qubits $\{Q_0,Q_1,Q_2,Q_3\}$ of the noisy simulator.
As in the case of the Mermin polynomial, we use the elimination methods to cancel the effect of quantum noise
and adopt the least square method to mitigate the impact of classical noise.

The experimental results are summarized in Figure~\ref{fig:GHZ state}.
We tell from the significant gap between the raw data (with and without error mitigation) 
and theoretical curves that the effect of quantum noise is quite significant.
We also tell from the minor gaps between the eliminated and theoretical curves that
all three elimination methods collaborate pretty well with the least square mitigation method.
They together remarkably improve the computation accuracy.
What's more, we conclude from the minor estimation errors
there is no statistical difference among these elimination methods.

\begin{figure}[!htbp]
    \centering
    \includegraphics[width=0.49\textwidth]{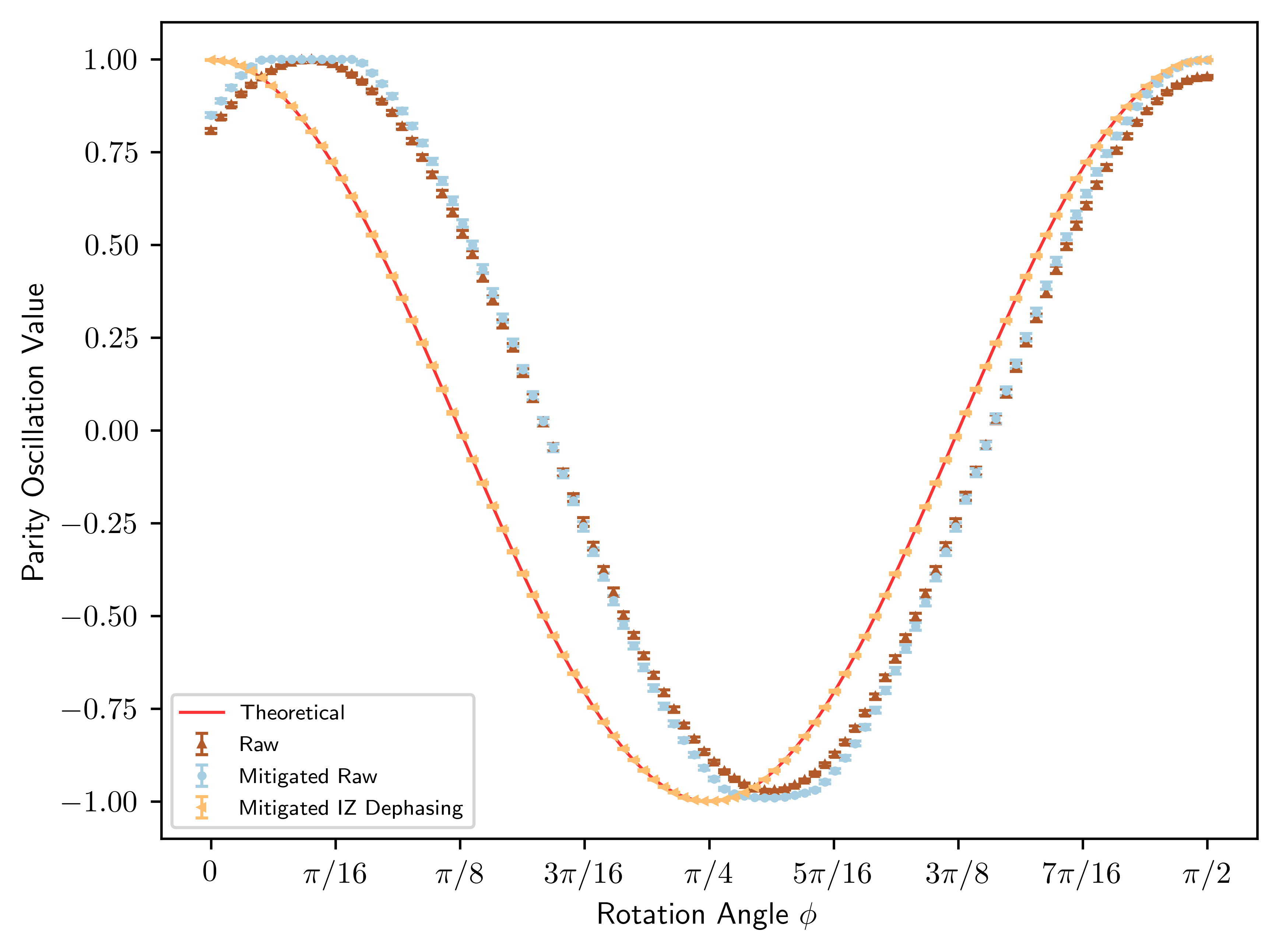}
    \caption{\raggedright
            Parity oscillation of the four-qubit GHZ state~\eqref{eq:GHZ}
            measured on Baidu Quantum Platform.
            Experimental data are reported as the average value over $10^2$ 
            independent estimations with standard error.
            \textbf{Theoretical} values (red line) are analytically calculated.
            \textbf{Raw} (brown upper triangle) and \textbf{Mitigated Raw} (blue circle) 
            are collected without or with the least square error mitigation.
            \textbf{Mitigated IZ dephasing} (yellow left triangle) values are collected by first eliminating quantum noise
            using the corresponding method and then performing least square error mitigation.}
    \label{fig:GHZ state}
\end{figure}

\subsection{Ground state energy of hydrogen molecule}

Our last application is to estimate the ground state energy of hydrogen molecules by running the
flagship algorithm Variational Quantum Eigensolver (VQE)~\cite{peruzzo2014variational} on the noisy simulator.
Briefly speaking, the inputs to a VQE algorithm are a Hamiltonian $H$ of the hydrogen molecule
and a parametrized circuit that prepares a trial state $\ket{\psi(\bm{\theta})}$ aiming to 
approximate the ground state of the molecule.
Within VQE, the cost function is defined to be the expectation value of the Hamiltonian 
computed in the trial state $f(\bm{\theta}):=\bra{\psi(\bm{\theta})}H\ket{\psi(\bm{\theta})}$,
which can be estimated by measuring $\psi(\bm{\theta})$.
The ground state of the target Hamiltonian is obtained by performing an iterative cost function minimization.
The optimization is carried out by a classical optimizer which leverages a quantum computer to evaluate the cost function and calculate its gradient at each optimization step.
The $4$-qubit Hamiltonian of hydrogen molecule we apply is
obtained from the OpenFermion library~\cite{mcclean2020openfermion}
and the exact form is given in Appendix~\ref{appx:vqe-configuration}.
We construct a $4$-qubit variational ansatz $\ket{\psi(\bm{\theta})}$
consisting of an initial state $\ket{0}^{\otimes 4}$ followed by $6$
repetitions of a layer of parameterized  single-qubit $Y$-rotations on each qubit and a layer of CZ gates between alternating qubits.
The parameters $\bm{\theta}$ are randomly initialized
and the classical optimizer is chosen as the sequential minimal optimization (SMO) method.

\begin{figure}[!htbp]
  \centering
  \includegraphics[width=0.48\textwidth]{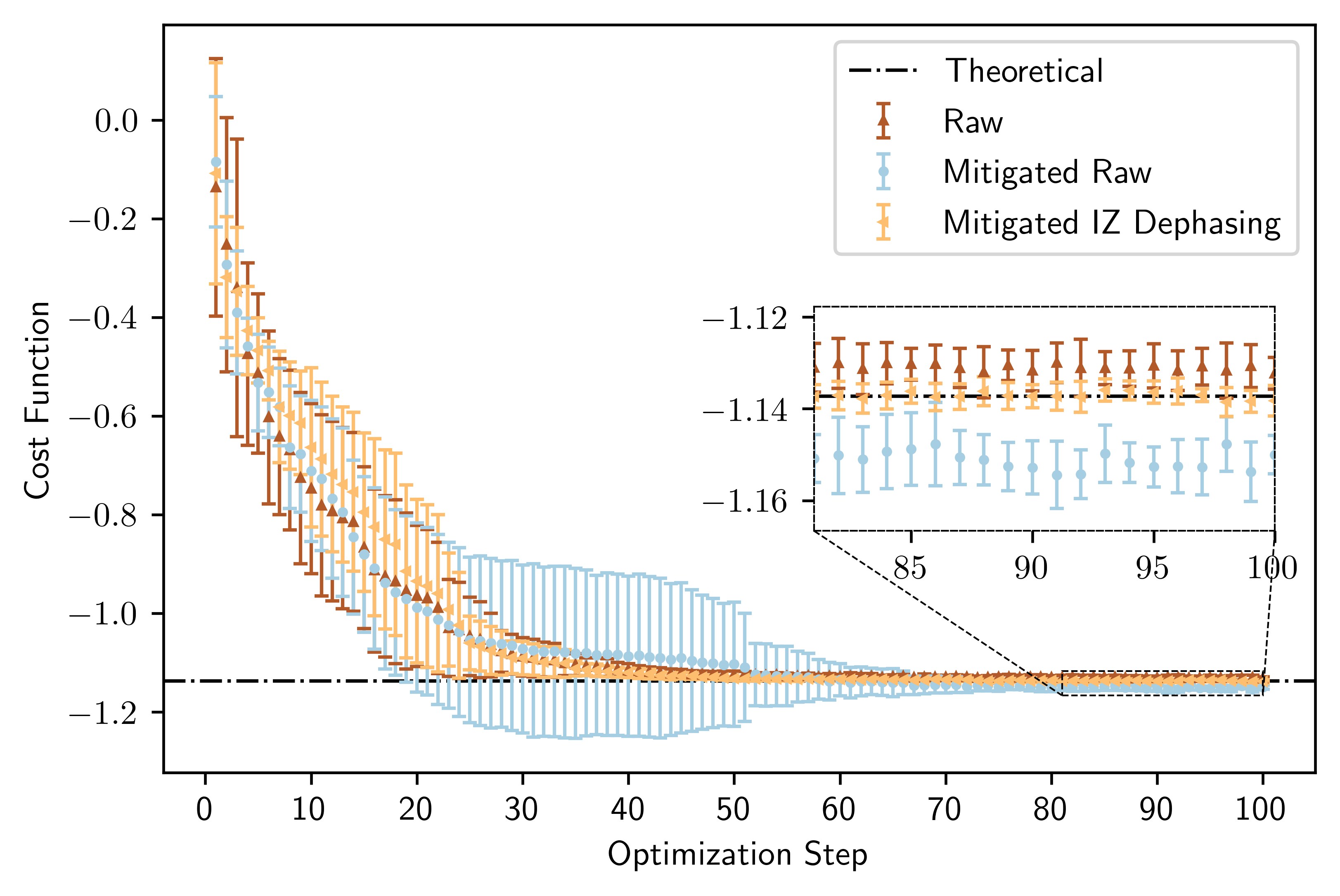}
  \caption{\raggedright
          The ground state energy of hydrogen molecule measured on Baidu Quantum Platform.
          Experimental data are reported as the average value over $10$ independent estimations with standard error.
          \textbf{Theoretical} (black dashed line) gives the theoretical ground energy $- 1.137$.
          \textbf{Raw} (brown upper triangle) and \textbf{Mitigated Raw} (blue circle) values
          are collected without or with the least square error mitigation.
          \textbf{Mitigated IZ dephasing} (yellow left triangle) values are 
          collected by first eliminating quantum noise 
          with IZ dephasing and then performing least square error mitigation.}
  \label{fig:VQE}
\end{figure}

The experimental results are summarized in Figure~\ref{fig:VQE}. 
We need to point out that we did not use the XY and Pauli twirling methods because they achieve the 
same performance as the IZ dephasing method as revealed in previous applications.
The \textbf{Mitigated IZ dephasing} data shows that quantum noise elimination combined with error mitigation 
can greatly improve the effectiveness of VQE algorithms, yields an accurate estimation of the ground stat energy.
On the other hand, we prove in Appendix~\ref{appx:vqe-configuration} that the \textbf{Raw} data actually estimates
the ground state energy of a transformed Hamiltonian whose theoretical value should be $- 1.135$.
The data matches this theoretical value pretty well.
When we execute the error mitigation method in the presence of quantum noise,
we see from \textbf{Mitigated Raw} that we obtain an overestimation ($- 1.152$) of the ground state energy.
This indicates that we should be cautious when performing measurement error mitigation methods:
if the measurement contains quantum noise, mitigation would be harmful instead of useful.

We find from the slow convergence rate of the raw cost function that VQE is immensely sensitive to measurement noise.
In the absence of noise elimination, the noisy cost function deviates significantly from its noiseless
counterpart, which can greatly limit the effectiveness of VQE algorithms, even for a small number of qubits.
We also find that the eliminated and mitigated cost values converge much faster and are slightly 
more accurate than raw ones. 

\section{Conclusions}\label{sec:conclusions}

The main contribution of this work is a two-stage procedure
that systematically addresses quantum noise inherent in NISQ measurement devices.
The procedure is incredibly intuitive: we first detect and then eliminate quantum noise if there is any.
In the first stage, we prepared maximally coherent states $\ket{\Phi_\theta}$ with relative phase $\theta$
and maximally mixed states as inputs to the measurement device and fitted the difference between two measurement statistics to the Fourier series.
The fitting coefficients quantitatively benchmark the quantum noise of the measurement device.
In the second stage, we executed randomly sampled Pauli gates before the measurement device. We conditionally flipped the outcomes
so that the resulting effective measurement contains only classical noise.
We demonstrated the procedure's practicability numerically on Baidu Quantum Platform,
via two paradigmatic quantum applications.
Remarkably, these results revealed that quantum noise in the measurement
devices under investigation are significantly suppressed,
and the computation accuracy of the applications is substantially improved.
This two-stage procedure complements existing measurement error mitigation techniques,
and we believe that they together form a standard toolbox for manipulating measurement errors in near-term quantum devices.

We have seen the devastating impact of quantum noise when mitigating measurement noises. 
This motivates us to study quantum noise from a resource theoretic perspective,
possibly termed a \emph{resource theory of decoherence}, which may
revolutionize how we understand and manipulate noises inherent in quantum measurement devices.
Also, it is worth discovering if there are more efficient ways to eliminate quantum noise other than the methods inspired by Pauli twirling.

\textbf{\textit{Acknowledgements.}} S. T. and C. Z. contributed equally to this work.
This work was done when S. T. and C. Z. were research interns at Baidu Research.
We would like to thank Runyao Duan for helpful discussions.


%


\begin{appendices}
\onecolumngrid

\section{Alternative definition of quantum noise witness}\label{appx:quantum noise witness}

\newtheorem{innercustomgeneric}{\customgenericname}
\providecommand{\customgenericname}{}
\newcommand{\newcustomtheorem}[2]{%
  \newenvironment{#1}[1]
  {%
   \renewcommand\customgenericname{#2}%
   \renewcommand\theinnercustomgeneric{##1}%
   \innercustomgeneric
  }
  {\endinnercustomgeneric}
}

\newcustomtheorem{customdef}{Definition}

Let's consider the following alternative definition of quantum noise witness,
where the separation hyperplane is determined by the expectation values less than or equal to $0$.
\begin{customdef}{2'}
Let $W$ be a Hermitian operator in $\cH_n$. $W$ is called a quantum noise witness, if
\begin{enumerate}
    \item for arbitrary classical POVM $\{ E^c_{\bm{x}}\}_{\bm{x}}$,
            it holds for arbitrary $\bm{x}$ that $\tr[W E^c_{\bm{x}}] \leq 0$;
    \item there exists at least
            one quantum POVM $\{ E^q_{\bm{x}}\}_{\bm{x}}$
            such that there exists some $\bm{x}$ for which $\tr[W E^q_{\bm{x}}] > 0$.
\end{enumerate}
Thus, if one measures $\tr[W E_{\bm{x}}]<0$ for some $ E_{\bm{x}}$,
one knows for sure that this POVM element, and the corresponding POVM, contains quantum noise.
\end{customdef}

Let $W$ be a quantum noise witness defined w.r.t. \textbf{Definition 2'}
and let $\cS$ be the set of quantum POVM elements that can be detected by $W$, i.e.,
\begin{align}
    \forall  E_{\bm{x}}^q \in \cS, \tr[W E_{\bm{x}}^q] > 0.
\end{align}
We will construct a quantum noise witness $\widetilde{W}$ from $W$ satisfying the following two conditions:
1) for arbitrary classical POVM element $ E^c_{\bm{x}}$, $\tr[\widetilde{W} E^c_{\bm{x}}] = 0$; and
2) $\forall  E_{\bm{x}}^q \in \cS$, $\tr[\widetilde{W} E_{\bm{x}}^q]\neq0$.
That is to say, given arbitrary witness defined w.r.t. \textbf{Definition 2'},
we can always construct a new witness which satisfies \textbf{Definition 2} in the main text
and the quantum POVM elements that can be detected by the former can also be detected
by the latter. Since a witness defined w.r.t. \textbf{Definition 2} naturally satisfies \textbf{Definition 2'},
we conclude that these two definitions are equivalent.

The construction is as follows. Define the following operator
\begin{align}
    \widetilde{W} := W - D_W,
\end{align}
where $D_W$ is the diagonal part of $W$. Note that $\widetilde{W}$ is Hermitian whenever $W$ is.
By definition it is easy to see that $\widetilde{W}$ only has non-zero off-diagonal elements.
For an arbitrary classical POVM element $ E_{\bm{x}}^c$, it holds that
\begin{align}
  \tr[\widetilde{W} E_{\bm{x}}^c]
= \tr[W E_{\bm{x}}^c] - \tr[D_W E_{\bm{x}}^c] = 0,
\end{align}
since $ E_{\bm{x}}^c$ only has non-zero diagonal elements.
On the other hand, for arbitrary quantum POVM element $ E_{\bm{x}}^q\in\cS$, it holds that
\begin{align}
  \tr[\widetilde{W} E_x^q]
= \tr[W E_x^q] - \tr[D_W E_x^q]
= \tr[W E_x^q] - \tr[WD_{ E_x^q}],
\end{align}
where $D_{ E_x^q}$ is the diagonal part of $ E_x^q$.
Notice that $\tr[W E_x^q]>0$ since $ E_{\bm{x}}^q\in\cS$
and $\tr[WD_{ E_x^q}]\leq0$ since $D_{ E_x^q}$ is a classical POVM
and $W$ is defined w.r.t. \textbf{Definition 2'}.
Therefore, $\tr[\widetilde{W} E_x^q]\neq 0$.

\section{Proof of Proposition~\ref{prop:completeness}}\label{appx:completeness}

Assume the spectral decomposition $W_{\psi} = \sum_{i}\lambda_i\proj{i}$, where
$W_\psi$ is a quantum noise witness defined in~\eqref{eq:psi-induced-witness-2}
and $\{\ket{i}\}_i$ forms an orthonormal basis.
For arbitrary probe state $\psi$, the density matrix can be written as $\rho = \proj{\psi}$. 
Assume $\ket{\psi} = \sum_{i} \alpha_i \ket{i}$, and

\begin{align}
    W_{\psi} &= \frac{\1}{2^n} - \proj{\psi} = \frac{\1}{2^n} - \sum_{i, j}\alpha_i\alpha_j\ketbra{i}{j}.
\end{align}
Therefore,

\begin{align}
    \tr[W_{\psi}  E_{\bm{x}}^q] &= \tr[\frac{\1}{2^n}D_{ E_{\bm{x}}^q}]-\sum_{i, j}\alpha_i\alpha_j\bra{j} E_{\bm{x}}^q\ket{i}\\
    &=\sum_{i}\left(\frac{1}{2^n}-\vert{\alpha_i}\vert^{2}\right)\bra{i} E_{\bm{x}}^q\ket{i} - \sum_{i\neq j}\alpha_i\alpha_j\bra{j} E_{\bm{x}}^q\ket{i}.\label{appx:completeness proof}
\end{align}

As long as $\vert{\alpha_i}\vert^{2}=\frac{1}{2^n}$, Eq.~\eqref{appx:completeness proof} would give $- \sum_{i\neq j}\alpha_i\alpha_j\bra{j} E_{\bm{x}}^q\ket{i}$. When $\tr[W_{\psi}  E_{\bm{x}}^q]\neq 0$, we can always say that there exists quantum noise.

\section{Properties of quantum noise measure}\label{sec:properties of quantum noise measure}

In this section, we explore general properties, analytical solutions in the qubit case,
and the relations to the resource theory of quantum measurements
of the $\psi$-induced quantum noise measure $\mathscr{Q}_\psi$ in
Definition~\ref{def:psi induced quantum noise measure} of the main text.

\subsection{General properties}

First of all, we show that for arbitrary POVM element $ E_{\bm{x}}$, the quantity
$\mathscr{Q}_\psi( E_{\bm{x}})$ is normalized.
\begin{lemma}[Normalization]\label{lemma:measure-range}
Let $ E_{\bm{x}}$ be an arbitrary POVM element. It holds that
    $\mathscr{Q}_\psi( E_{\bm{x}})\in \left[0, 1\right]$.
\end{lemma}
\begin{proof}
First we show that for arbitrary positive semidefinite matrix $A$ it holds that
\begin{align}\label{eq:pWkMlMghnAjpZs}
    \vert A_{i,j} \vert \leq \sqrt{A_{i,i}A_{j,j}},
\end{align}
where $A_{i,j}$ is the element in the $i$-th row and $j$-th column.
For a positive semidefinite matrix $A$, all of its principal submatrices are positive semidefinite.
Furthermore, $\tr{A}$, $\det{A}$ as well as the principal minors of $A$ are all
nonnegative~\cite{horn2012matrix}. Consider the following $2\times 2$ principal minor
\begin{align}
    \wh{A} := \begin{bmatrix} A_{i,i} & A_{i,j}\\ A_{i,j}^\dagger & A_{jj} \end{bmatrix}.
\end{align}
It holds that
\begin{align}
    \det(\wh{A})
= A_{i,i}A_{j,j}-\vert A_{i,j}\vert^2\geq0,
\end{align}
which leads to~\eqref{eq:pWkMlMghnAjpZs}.

For arbitrary POVM element $ E_{\bm{x}}$, since it is positive semidefinite,
we obtain from Eq.~\eqref{eq:pWkMlMghnAjpZs} that
\begin{align}
    \vert  E_{\bm{x}}(\bm{y},\bm{z})\vert
\leq \sqrt{ E_{\bm{x}}(\bm{y},\bm{y}) E_{\bm{x}}(\bm{z},\bm{z})}. \label{povm-1}
\end{align}
Therefore,
\begin{align}
    \left\vert\sum_{\bm{y}\neq\bm{z}}  E_{\bm{x}}(\bm{y},\bm{z})\right\vert
\leq \sum_{\bm{y}\neq\bm{z}} \left\vert E_{\bm{x}}(\bm{y},\bm{z})\right\vert
\leq \sum_{\bm{y}\neq\bm{z}} \sqrt{ E_{\bm{x}}(\bm{y},\bm{y}) E_{\bm{x}}(\bm{z},\bm{z})}. \label{povm-2}
\end{align}
On the other hand, since $\sum_{\bm{x}} E_{\bm{x}}=I$ and $ E_{\bm{x}}\geq0$,
it follows that $0\leq E_{\bm{x}}(\bm{y},\bm{y})\leq1$ for all $\bm{y}$. Thus
\begin{align}
    \sqrt{ E_{\bm{x}}(\bm{y},\bm{y}) E_{\bm{x}}(\bm{z},\bm{z})}\leq \sum_{\bm{y}\neq\bm{z}}\sqrt{ E_{\bm{x}}(\bm{y},\bm{y}) E_{\bm{x}}(\bm{z},\bm{z})}\leq \sum_{\bm{y}\neq\bm{z}} E_{\bm{x}}(\bm{y},\bm{y}) E_{\bm{x}}(\bm{z},\bm{z})\leq\sum_{\bm{y}} E_{\bm{x}}(\bm{y},\bm{y})=1.
\end{align}

Finally, given
\begin{align}
    0\leq\vert e^{i(\theta_{\bm{z}} - \theta_{\bm{y}})}\vert\leq1,
\end{align}
it's obvious that
\begin{align}
    \mathscr{Q}_\psi( E_{\bm{x}})=\vert\sum_{\bm{y}\neq\bm{z}} e^{i(\theta_{\bm{z}} - \theta_{\bm{y}})} E_{\bm{x}}(\bm{y},\bm{z})\vert\leq1.
\end{align}
\end{proof}

\begin{lemma}[Additivity] \label{lemma:additivity}
Let $\bm{ E}^{(1)}=\{ E^{(1)}_{\bm{x}}\}_{\bm{x}\in\{0,1\}^n}$
and $\bm{ E}^{(2)}=\{ E^{(2)}_{\bm{x}}\}_{\bm{x}\in\{0,1\}^n}$
be two POVMs. For arbitrary $\bm{x}\in\{0,1\}^n$, it holds that
\begin{align}
  \mathscr{Q}_\psi( E_{\bm{x}}^{(1)}\otimes E_{\bm{x}}^{(2)})
\leq 2(\mathscr{Q}_\psi( E_{\bm{x}}^{(1)}) + \mathscr{Q}_\psi( E_{\bm{x}}^{(2)})). \label{eq:additivity}
\end{align}
Correspondingly,
\begin{align}
  \mathscr{Q}_\psi(\bm{ E}^{(1)}\otimes\bm{ E}^{(2)})
\leq 2(\mathscr{Q}_\psi(\bm{ E}^{(1)}) + \mathscr{Q}_\psi(\bm{ E}^{(2)})).
\end{align}
\end{lemma}

\begin{proof}
Consider $2$ qubits case. Assume that for qubit $0$ and qubit $1$,

\begin{align}
     E_0^{(0)}= \left[\begin{matrix}
    a_0 & \gamma_0 \\
    \gamma_0^\ast & b_0
    \end{matrix}\right],\quad\quad
     E_1^{(0)}= \left[\begin{matrix}
    1-a_0 & -\gamma_0 \\
    -\gamma_0^\ast & 1-b_0
    \end{matrix}\right],
\end{align}

\begin{align}
     E_0^{(1)}= \left[\begin{matrix}
    a_1 & \gamma_1 \\
    \gamma_1^\ast & b_1
    \end{matrix}\right],\quad\quad
     E_1^{(1)}= \left[\begin{matrix}
    1-a_1 & -\gamma_1 \\
    -\gamma_1^\ast & 1-b_1
    \end{matrix}\right].
\end{align}

We can see that
\begin{align}
    \mathscr{Q}_\Phi( E^{(0)}_{0}) + \mathscr{Q}_\Phi( E^{(1)}_{0}) = 2\left\vert\Re(\gamma_0)\right\vert + 2\left\vert\Re(\gamma_1)\right\vert.
\end{align}

Then,

\begin{align}
     E_{00} =  E_0^{(0)} \otimes  E_1^{(1)} = \left[\begin{matrix}
    a_0 E_0^{(1)} & \gamma_0 E_0^{(1)} \\
    \gamma_0^\ast  E_0^{(1)} & b_0 E_0^{(1)}
    \end{matrix}\right]. \label{eq:quantum noise-2qubits}
\end{align}

Expand Eq.~\eqref{eq:quantum noise-2qubits} we get

\begin{align}
     E_{00} = \left[\begin{matrix}
    a_0a_1 & a_0\gamma_1 & \gamma_0a_1 & \gamma_0\gamma_1 \\
    a_0\gamma^\ast_1 & a_0b_1 & \gamma_0\gamma_1^\ast & \gamma_0b_1 \\
    \gamma_0^\ast a_1 & \gamma_0^\ast\gamma_1 & b_0a_1 & b_0\gamma_1 \\
    \gamma_0^\ast \gamma_1^\ast & \gamma_0^\ast b_1 & b_0\gamma_1^\ast & b_0b_1
    \end{matrix}\right].
\end{align}

Then
\begin{align}
    \mathscr{Q}_\Phi( E_{00}) =2 \left\vert (a_0+b_0)\Re(\gamma_1) + (a_1+b_1)\Re(\gamma_0) + 2\Re(\gamma_0)\Re(\gamma_1)\right\vert. \label{multi-case}
\end{align}

Let's consider the maximum case. Assume $a_0 = a_1 = 1$, $b_0 = b_1 = 1$. Eq.~\eqref{multi-case} turns into
\begin{align}
    \max(\mathscr{Q}_\Phi( E_{00}))&=4 \left\vert \Re(\gamma_1) + \Re(\gamma_0) + \Re(\gamma_0\gamma_1) + \Re(\gamma_0\gamma_1^\ast) \right\vert \\
    &=4 \left\vert \Re(\gamma_1) + \Re(\gamma_0) + 2\Re(\gamma_0)\Re(\gamma_1) \right\vert.
\end{align}

Through triangle inequality,
\begin{align}
    \left\vert \Re(\gamma_1) + \Re(\gamma_0) + 2\Re(\gamma_0)\Re(\gamma_1) \right\vert &\leq \left\vert \Re(\gamma_1) + \Re(\gamma_0)\vert + \vert 2\Re(\gamma_0)\Re(\gamma_1) \right\vert \\
    &\leq \left\vert \Re(\gamma_1) + \Re(\gamma_0) \right\vert \\
    &\leq \left\vert \Re(\gamma_1) \right\vert + \left\vert \Re(\gamma_0) \right\vert.
\end{align}

Therefore,
\begin{align}
    \max(\mathscr{Q}_\Phi( E_{00})) \leq 2(\mathscr{Q}_\Phi( E^{(0)}_{0}) + \mathscr{Q}_\Phi( E^{(1)}_{0})).
\end{align}

 Without loss of generality, we can conclude that

 \begin{align}
  \mathscr{Q}_\psi( E_{\bm{x}}^{(1)}\otimes E_{\bm{x}}^{(2)})
\leq 2(\mathscr{Q}_\psi( E_{\bm{x}}^{(1)}) + \mathscr{Q}_\psi( E_{\bm{x}}^{(2)})).
\end{align}
Correspondingly,
\begin{align}
  \mathscr{Q}_\psi(\bm{ E}^{(1)}\otimes\bm{ E}^{(2)})
\leq 2(\mathscr{Q}_\psi(\bm{ E}^{(1)}) + \mathscr{Q}_\psi(\bm{ E}^{(2)})).
\end{align}
\end{proof}

\begin{lemma}[Convexity] \label{lemma:convexity}
Let $\bm{ E}^{(1)}=\{ E^{(1)}_{\bm{x}}\}_{\bm{x}\in\{0,1\}^n}$
and $\bm{ E}^{(2)}=\{ E^{(2)}_{\bm{x}}\}_{\bm{x}\in\{0,1\}^n}$
be two POVMs. For arbitrary $p\in[0,1]$ and $\bm{x}\in\{0,1\}^n$,
it holds that
\begin{align}
     \mathscr{Q}_\psi(p E_{\bm{x}}^{(1)} + (1-p) E_{\bm{x}}^{(2)})
\leq p\mathscr{Q}_\psi( E_{\bm{x}}^{(1)}) + (1-p)\mathscr{Q}_\psi( E_{\bm{x}}^{(2)}).
\end{align}
Correspondingly,
\begin{align}
     \mathscr{Q}_\psi(p\bm{ E}^{(1)}+(1-p)\bm{ E}^{(2)})
\leq p\mathscr{Q}_\psi(\bm{ E}^{(1)}) + (1-p)\mathscr{Q}_\psi(\bm{ E}^{(2)}).
\end{align}
\end{lemma}
\begin{proof}
For arbitrary $p\in[0,1]$, we have
\begin{align}
    \mathscr{Q}(p E_{\bm{x}}^{(1)} +(1-p) E_{\bm{x}}^{(2)} )&=2^n\left\vert \tr[W_{\psi}(p E_{\bm{x}}^{(1)} + (1-p)p E_{\bm{x}}^{(2)})] \right\vert \\
    &=2^n\left\vert \tr[p W_{\psi} E_{\bm{x}}^{(1)}] + \tr[(1-p) W_{\psi} E_{\bm{x}}^{(2)}] \right\vert \\
    &=2^n\left\vert p\tr[W_{\psi} E_{\bm{x}}^{(1)}] + (1-p)\tr[W_{\psi} E_{\bm{x}}^{(2)} ] \right\vert
\end{align}
and
\begin{align}
    p\mathscr{Q}( E_{\bm{x}}^{(1)}) + (1-p)\mathscr{Q}( E_{\bm{x}}^{(2)}) =
    2^n\left\vert p\tr[W_{\psi} E_{\bm{x}}^{(1)}] \right\vert + 2^n\left\vert (1-p)\tr[W_{\psi} E_{\bm{x}}^{(2)}] \right\vert.
\end{align}
From triangle inequality,
\begin{align}
    \left\vert p\tr[W_{\psi} E_{\bm{x}}^{(1)}] + (1-p)\tr[W_{\psi} E_{\bm{x}}^{(2)}]
    \right\vert \leq
    \left\vert p\tr[W_{\psi} E_{\bm{x}}^{(1)}] \right\vert + \left\vert (1-p)\tr[W_{\psi} E_{\bm{x}}^{(2)}] \right\vert.
\end{align}
Therefore,
\begin{align}
  \mathscr{Q}(p E_{\bm{x}}^{(1)} + (1-p) E_{\bm{x}}^{(2)})
\leq p\mathscr{Q}( E_{\bm{x}}^{(1)})
    + (1-p)\mathscr{Q}( E_{\bm{x}}^{(2)}).
\end{align}
\end{proof}

\subsection{Single and two-qubit case}\label{qubit case}

Let's begin with the single-qubit case. Consider a qubit POVM $\{ E_0, E_1\}$ with
\begin{align}
     E_0
= \begin{bmatrix}
    \alpha & a_1+i b_1 \\a_1-i b_1 & \beta\\
    \end{bmatrix}, \qquad
     E_1
= \begin{bmatrix}
    1-\alpha & - a_1 - i b_1 \\
    - a_1 + i b_1 & 1-\beta
  \end{bmatrix},
\end{align}
where $\alpha,\beta,a_1,b_1\in\mathbb{R}$. We have
\begin{align}
    \mathscr{Q}^\theta_{\Phi}( E_0) = 2\left\vert -a_1\cos(\theta) + b_1\sin(\theta)\right\vert, \qquad
    \mathscr{Q}^\theta_{\Phi}( E_1) = 2\left\vert a_1\cos(\theta) - b_1\sin(\theta)\right\vert.
\end{align}
It is obvious that we just need to consider $2$ special cases in order to extract
the absolute values of $a_1$ and $b_1$. Specifically,
\begin{align}
    \mathscr{Q}^{\theta=0}_{\Phi}( E_0)
&= \mathscr{Q}^{\theta=0}_{\Phi}( E_1) = 2\left\vert a_1\right\vert, \\
    \mathscr{Q}^{\theta=\pi/2}_{\Phi}( E_0)
&= \mathscr{Q}^{\theta=\pi/2}_{\Phi}( E_1) = 2\left\vert b_1\right\vert.
\end{align}

Now we consider the two-qubit case.
The density matrix of two-qubit $\Phi_\theta$~\eqref{eq:Phi-theta} has the form
\begin{align}
    \Phi_{\theta} =\frac{1}{4}\left[
    \begin{array}{cccc}
    1 & e^{-i\theta} & e^{-i\theta} & e^{-i2\theta} \\
    e^{i\theta} & 1 & 1 & e^{-i\theta} \\
    e^{i\theta} & 1 & 1 & e^{-i\theta} \\
    e^{i2\theta} & e^{i\theta} & e^{i\theta} & 1
    \end{array}
    \right].
\end{align}
Suppose we have an unknown POVM element $ E$ parameterized as
\begin{align}
     E=\left[
    \begin{array}{cccc}
    a_{11} & a_{12} & a_{13} & a_{14} \\
    a_{12}^\ast & a_{22} & a_{23} & a_{24} \\
    a_{13}^\ast & a_{23}^\ast & a_{33} & a_{34} \\
    a_{14}^\ast & a_{24}^\ast & a_{34}^\ast & a_{44}
    \end{array}
    \right].
\end{align}
After some tedious calculation we obtain
\begin{align}
    \mathscr{Q}_{\Phi}^{\theta}( E)
&= 2\left\vert-\cos{(\theta)}\Re({a_{12}+a_{13}+a_{24}+a_{34}}) +\sin{(\theta)}\Im{(a_{12}+a_{13}+a_{24}+a_{34}})\right. \notag\\
&\left.\quad\quad- \cos{(2\theta)}\Re{(a_{14})}+\sin{(2\theta)}\Im{(a_{14})}-\Re{(a_{23})}\right\vert \\
&= 2\left\vert -a_0 -a_1 \cos{(\theta)} -a_2 \cos{(2\theta)} + b_1 \sin{(\theta)} + b_2 \sin{(2\theta)}\right\vert ,
\end{align}
where
\begin{align}
    a_0 &= \Re{(a_{23})}, \\
    a_1 &= \Re({a_{12}+a_{13}+a_{24}+a_{34}}), \\
    a_2 &= \Re{(a_{14})}, \\
    b_1 &= \Im{(a_{12}+a_{13}+a_{24}+a_{34}}), \\
    b_2 &= \Im{(a_{14})}.
\end{align}

\subsection{Relation to quantum resource theory}

In~\cite{Baek_2020}, the authors formalized a resource theoretical framework
to quantify the coherence of quantum measurements
and introduced a coherence monotone of measurement in terms of
off-diagonal values of the POVM elements.
Specifically, Let $\bm{ E}=\{ E_{\bm{x}}\}_{\bm{x}}$ be a POVM,
the $\ell_\infty$-coherence monotone of $\bm{ E}$ is
defined as~\cite[Eq. (10)]{Baek_2020}
\begin{align}
    \mathcal{C}_{\ell_\infty}(\bm{ E})
:=  \sum_{\bm{x}} \sum_{\bm{y}<\bm{z}}
    \left\vert  E_{\bm{x}}(\bm{y},\bm{z})\right\vert.
    \label{eq:coherence-measure}
\end{align}

We can show that the quantum noise measure $\mathscr{Q}_\Phi$ bounds
the $\ell_\infty$-coherence monotone from below.
This result, together with Theorem 1 in~\cite{Baek_2020},
establishes an experimental-friendly lower bound on the
critical quantum measurement robustness measure
originally introduced in~\cite{Baek_2020}.

\begin{lemma}\label{lemma:lower-bound}
Let $\bm{ E}=\{ E_{\bm{x}}\}_{\bm{x}}$ be a POVM in $\cH_n$.
It holds that
\begin{align}
    \mathcal{C}_{\ell_\infty}(\bm{ E}) \geq 2^{n-1}\mathscr{Q}_\Phi^{\theta=0}(\bm{ E}).\label{eq:lower-bound}
\end{align}
\end{lemma}
\begin{proof}
By definition~\eqref{eq:psi-quantum noise measure 2}, we have
\begin{align}
    \mathscr{Q}_\Phi^{\theta=0}(\bm{ E})
&= \frac{1}{2^n}\sum_{\bm{x}}\mathscr{Q}_\Phi^{\theta=0}( E_{\bm{x}}) \\
&= \frac{1}{2^{n-1}} \sum_{\bm{x}} \left\vert \sum_{\bm{y}<\bm{z}} E_{\bm{x}}(\bm{y}, \bm{z}) \right\vert \\
&\leq \frac{1}{2^{n-1}} \sum_{\bm{x}} \sum_{\bm{y}<\bm{z}}\left\vert  E_{\bm{x}}(\bm{y}, \bm{z}) \right\vert \\
&= \frac{1}{2^{n-1}}\mathcal{C}_{\ell_\infty}(\bm{ E}),
\end{align}
leading to~\eqref{eq:lower-bound}.
\end{proof}

\section{Proof of Theorem~\ref{thm:Phi-Psi-quantum noise}}\label{appx:Phi-Psi-quantum noise}

Notice that the prob state $\Phi_\theta$ defined in~\eqref{eq:Phi-theta} of the main text
can be equivalently expressed as
\begin{align}
    \ket{\Phi_\theta} = \frac{1}{\sqrt{2^n}} \sum_{\bm{y}}e^{i\theta \left\vert \bm{y}\right\vert}\ket{\bm{y}},
\end{align}
where $\vert\bm{y}\vert$ represents the Hamming weight of the binary string $\bm{y}$.
Thus
\begin{align}
    \Phi_\theta
\equiv \proj{\Phi_\theta}
= \frac{1}{2^n}\sum_{\bm{y, z}}e^{i\theta (\left\vert \bm{y}\right\vert-\left\vert \bm{z}\right\vert)}
        \ketbra{\bm{y}}{\bm{z}}.
\end{align}
According to Eq.~\eqref{eq:difference}, it holds that
\begin{align}
    \mathscr{Q}_{\Phi}^{\theta}( E_{\bm{x}})
&= 2^n\left\vert\tr[W_{\Phi} E_{\bm{x}}] \right\vert \\
    &= 2^n\left\vert\tr\left[\left(\frac{1}{2^n}\1 - \Phi_\theta\right) E_{\bm{x}}\right]\right\vert \\
    &= \left\vert\tr\left[\left(\1 - \sum_{\bm{y, z}}e^{i\theta(\left\vert \bm{y}\right\vert-\left\vert \bm{z}\right\vert)}
                \ketbra{\bm{y}}{\bm{z}}\right) E_{\bm{x}}\right]\right\vert \\
    &= \left\vert\sum_{\bm{y}\neq\bm{z}} \tr[( - e^{i\theta (\left\vert \bm{y}\right\vert-\left\vert \bm{z}\right\vert)}
            \ketbra{\bm{y}}{\bm{z}}) E_{\bm{x}}] \right\vert \\
    &= \left\vert -\sum_{\bm{y}\neq\bm{z}} e^{i\theta (\left\vert \bm{y}\right\vert-\left\vert \bm{z}\right\vert)}
             E_{\bm{x}}(\bm{y},\bm{z}) \right\vert \label{step-2}\\
    &= 2\left\vert-\sum_{\bm{y< z}}
        \Re\left[ E_{\bm{x}}(\bm{y},\bm{z})\right]
        \cos\left[\left(\left\vert\bm{y}\right\vert-\left\vert\bm{z}\right\vert \right)\theta\right]
      + \Im\left[  E_{\bm{x}}(\bm{y},\bm{z})\right]
        \sin\left[\left(\left\vert\bm{y}\right\vert-\left\vert\bm{z}\right\vert \right)\theta\right]
        \right\vert,\label{step-1}
\end{align}
where Eq.~\eqref{step-1} follows from the conjugated property of the off-diagonal values of a POVM element.

It is worth pointing out that $\bm{y}<\bm{z}$ does not
implies $\left\vert \bm{y} \right\vert < \left\vert \bm{z}\right\vert$ in general.
For example, when $\bm{y} = 011$ and $\bm{z} = 100$,
we can see that $\bm{y}<\bm{z}$ and $\left\vert \bm{y}\right\vert = 2 > \left\vert \bm{z}\right\vert = 1$.


\section{Proof of Proposition~\ref{prop:from-PTM-to-POVM}}\label{appx:prop:from-PTM-to-POVM}

First of all, we introduce some notations concerning the Pauli operators. 
The Pauli group $\mathsf{P}^n = \left\{I,X,Y,Z\right\}^{\otimes n}$ is Abelian and isomorphic to $\mathsf{Z}^{2n}_2$, 
where $\bm{a}\in\mathsf{Z}^{2n}_2$ is a $2n$ length binary string. 
Since $P_{\bm{a}}\in\mathsf{P}^n$ can be represented as 
\begin{align}
    P_{\bm{a}} = P_{(\bm{a}_x,\bm{a}_z)} = i^{\bm{a}_x\cdot\bm{a}_z}X[\bm{a}_x]Z[\bm{a}_z], 
\end{align}
where $\bm{a}_x,\bm{a}_z\in\mathsf{Z}^n_2$ and $A[\bm{a}] := \bigotimes_{\bm{a}_i\in\bm{a}}A^{\bm{a}_i}$. 
Given two Pauli operators $P_{\bm{a}}$ and $P_{\bm{b}}$, we have $P_{\bm{a}}P_{\bm{b}}=(-1)^{\langle\bm{a},\bm{b}\rangle}P_{\bm{b}}P_{\bm{a}}$ where 
\begin{align}
    \langle\bm{a},\bm{b}\rangle = \bm{a}_x\cdot\bm{b}_z+\bm{a}_z\cdot\bm{b}_x. 
\end{align}
In the following proof, we use rescaled Pauli operators $P_{\bm{a}}\to P_{\bm{a}}/\sqrt{2^n}$
so that the basis is properly normalized. 

Now let's prove Proposition~\ref{prop:from-PTM-to-POVM}. 
Notice that the PTM $[\cM]$ of a measurement channel $\cM$ satisfies
\begin{align}\label{eq:WEmMHi}
[\cM]_{\bm{i}\bm{j}} &= \tr[P_{\bm{i}}\cM(P_{\bm{j}})]
= \sum_{\bm{x}} \tr[ E_{\bm{x}}P_{\bm{j}}] \langle\bm{x}|P_{\bm{i}}|\bm{x}\rangle
= \begin{cases}
    \sum_{\bm{x}} \tr[ E_{\bm{x}}P_{\bm{j}}] \langle\bm{x}|P_{\bm{i}}|\bm{x}\rangle, &P_{\bm{i}}\in\{I,Z\}^{\otimes n}\\
    0, &P_{\bm{i}}\not\in\{I,Z\}^{\otimes n}
\end{cases}.
\end{align}
From this equation, we can see that an element of $[\cM]$ is non-zero if and only if
its row index corresponds to Pauli operators in the subset $\{I, Z\}^{\otimes n}$. 

We express Eq.~\eqref{eq:WEmMHi} in matrix form as
\begin{align}\label{eq:WEmMHi1}
    [\cM] = [\mathsf{P}][ E],
\end{align}
where $[\mathsf{P}]_{\bm{i}\bm{x}} := \langle\bm{x}|P_{\bm{i}}|\bm{x}\rangle$
and $[ E]_{\bm{x}\bm{j}} := \tr[ E_{\bm{x}}P_{\bm{j}}]$. Notice that
\begin{align}
    [\mathsf{P}]^\dagger[\mathsf{P}]
&=  \left(\sum_{\bm{j},\bm{y}}\langle\bm{y}|P_{\bm{j}}|\bm{y}\rangle\ketbra{\bm{y}}{\bm{j}}\right)\left(\sum_{\bm{i},\bm{x}}\langle\bm{x}|P_{\bm{i}}|\bm{x}\rangle\ketbra{\bm{i}}{\bm{x}}\right)\\ 
&=  \sum_{\bm{i},\bm{x},\bm{y}}\langle\bm{x}|P_{\bm{i}}|\bm{x}\rangle\langle\bm{y}|P_{\bm{i}}|\bm{y}\rangle\ketbra{\bm{y}}{\bm{x}}\\
&=  \sum_{\bm{x},\bm{y}}\ketbra{\bm{y}}{\bm{x}}\sum_{\bm{i}}\langle\bm{x}|P_{\bm{i}}|\bm{x}\rangle\langle\bm{y}|P_{\bm{i}}|\bm{y}\rangle\\
&=  \sum_{\bm{x},\bm{y}}\ketbra{\bm{y}}{\bm{x}}\left(\sum_{\bm{i}}\langle\bm{x}|Z[\bm{i}_z]|\bm{x}\rangle\langle\bm{y}|Z[\bm{i}_z]|\bm{y}\rangle\right)\\
&=  \sum_{\bm{x},\bm{y}}\ketbra{\bm{y}}{\bm{x}}\left(\frac{1}{2^n}\sum_{\bm{i}}(-1)\left[\bm{i}_z\cdot\bm{x}\right](-1)\left[\bm{i}_z\cdot\bm{y}\right]\right) \\
&=  \sum_{\bm{x},\bm{y}}\ketbra{\bm{y}}{\bm{x}}\left(\frac{1}{2^n}\sum_{\bm{i}}(-1)\left[\bm{i}_z\cdot(\bm{x}+\bm{y})\right]\right)\\
&=  \sum_{\bm{x},\bm{y}}\delta_{\bm{x},\bm{y}}\ketbra{\bm{y}}{\bm{x}} \\
&=  \1,
\end{align}
we obtain from Eq.~\eqref{eq:WEmMHi1} that
\begin{align}
    [ E] = [\mathsf{P}]^\dagger [\mathsf{P}][ E] = [\mathsf{P}]^\dagger[\cM].
\end{align}
This gives that for arbitrary $\bm{x}\in\{0,1\}^n$, 
\begin{align}
 E_{\bm{x}} = \sum_{P_{\bm{j}}\in\mathsf{P}^n}[ E]_{\bm{x}\bm{j}}P_{\bm{j}}
= \sum_{P_{\bm{i}},P_{\bm{j}}\in\mathsf{P}^n}[\mathsf{P}]_{\bm{i}\bm{x}}[\cM]_{\bm{i}\bm{j}}P_{\bm{j}}.
\label{eq:expand Pi_x}
\end{align}

\section{Proof of Eq.~(\ref{eq:IZ dephasing from twirling})}\label{appx:iz-dephasing}

Denote the RHS. of Eq.~\eqref{eq:IZ dephasing from twirling} as
\begin{align}
    \wt{\cM}(\cdot) := \frac{1}{2^n}\sum_{P\in\{I,Z\}^{\otimes n}} P\cM(P(\cdot)P)P.
\end{align}
We need to show that for arbitrary $\rho$ it holds that
\begin{align}
    \cM^{\operatorname{IZ}}(\rho) = \wt{\cM}(\rho).
\end{align}
Consider the following chain of equalities: 
\begin{align}
    \wt{\cM}(\rho)
&=  \frac{1}{2^n}\sum_{P\in\{I,Z\}^{\otimes n}} P\cM(P\rho P)P \\
&=  \frac{1}{2^n}\sum_{P\in\{I,Z\}^{\otimes n}} P
    \left(\sum_{\bm{x}\in\Sigma}\tr\left[ E_{\bm{x}}P\rho P\right]\proj{\bm{x}}\right)P \label{eq:tWIrjrRbz1} \\
&=  \frac{1}{2^n}\sum_{P\in\{I,Z\}^{\otimes n}}\sum_{\bm{x}\in\Sigma}\tr\left[ E_{\bm{x}}P\rho P\right]
        P\proj{\bm{x}}P \\
&=  \frac{1}{2^n}\sum_{P\in\{I,Z\}^{\otimes n}}\sum_{\bm{x}\in\Sigma}\tr\left[ E_{\bm{x}}P\rho P\right]\proj{\bm{x}}
    \label{eq:tWIrjrRbz2} \\
&=  \cM^{\operatorname{IZ}}(\rho),\label{eq:DTrCxDJBhWJBuPCDUx0}
\end{align}
where Eq.~\eqref{eq:tWIrjrRbz1} follows from the definition of $\cM$ in~\eqref{eq:measurement-channel}
and Eq.~\eqref{eq:tWIrjrRbz2} follows from the fact that for arbitrary $P\in\{I,X\}$, 
$P\proj{0}P=\proj{0}$ and $P\proj{1}P=\proj{1}$, 
i.e., $P\in\{I,Z\}^{\otimes n}$ preserves the classical state $\proj{\bm{x}}$. 
We are done. 

\section{Proof of Theorem~\ref{thm: XY twirling}}\label{appx:xy-twirling}

The PTM matrix of the XY-twirled measurement channel $\cM^{\operatorname{XY}}$ has the form
\begin{align}
    [\cM^{\operatorname{XY}}]_{\bm{i}\bm{j}}
&=  \tr\left[P_{\bm{i}}\cM^{\operatorname{XY}}(P_{\bm{j}})\right] \\
&=  \frac{1}{2^n}\sum_{P_{\bm{k}}\in\{X,Y\}^{\otimes n}}\tr\left[P_{\bm{i}} P_{\bm{k}} \cM(P_{\bm{k}}P_{\bm{j}}P_{\bm{k}}) P_{\bm{k}}\right] \\
&=  \left(\frac{1}{2^n}\sum_{P_{\bm{k}}\in\{X,Y\}^{\otimes n}}
        (-1)^{\langle\bm{i}+\bm{j},\bm{k}\rangle}\right)[\cM]_{\bm{i}\bm{j}},\label{eq:DteX3}
\end{align}
Note that
\begin{align}
    \frac{1}{2^n}\sum_{P_{\bm{k}}\in\{X,Y\}^{\otimes n}}(-1)^{\langle\bm{i}+\bm{j},\bm{k}\rangle}
=  \frac{1}{2^n}\sum_{P_{\bm{k}}\in\{X,Y\}^{\otimes n}}(-1)^{(\bm{i}+\bm{j})_x\cdot\bm{k}_z+(\bm{i}+\bm{j})_z\cdot\bm{k}_x}
=  \frac{(-1)^{(\bm{i}+\bm{j})_z}}{2^n}\sum_{P_{\bm{k}}\in\{X,Y\}^{\otimes n}}(-1)^{(\bm{i}+\bm{j})_x\cdot\bm{k}_z}.
\end{align}
Since we have shown in~\eqref{eq:WEmMHi} that an element of $[\cM]$ is non-zero if and only if
its row index corresponds to Pauli operators in the subset $\{I,Z\}^{\otimes n}$, 
we can assume $P_{\bm{i}}\in \{I, Z\}^{\otimes n}$. 
Now let's consider two cases of $P_{\bm{j}}$. 
If $P_{\bm{j}}\in\{I, Z\}^{\otimes n}$, it holds that
\begin{align}\label{eq:DteX1}
    \frac{1}{2^n}\sum_{P_{\bm{k}}\in\{X,Y\}^{\otimes n}}(-1)^{(\bm{i}+\bm{j})_x\cdot\bm{k}_z} = 1.
\end{align}
If $P_{\bm{j}}\not\in\{I, Z\}^{\otimes n}$, it holds that
\begin{align}\label{eq:DteX2}
    \frac{1}{2^n}\sum_{P_{\bm{k}}\in\{X,Y\}^{\otimes n}}(-1)^{(\bm{i}+\bm{j})_x\cdot\bm{k}_z}
= \frac{1}{2^n} \sum_{P_k \in \{X,Y\}^{\otimes n}} (-1)^{\bm{j}_x \cdot \bm{k}_z}
= 0.
\end{align}
Substituting Eqs.~\eqref{eq:DteX1} and~\eqref{eq:DteX2} into~\eqref{eq:DteX3}, we obtain
\begin{align}
    [\cM^{\operatorname{XY}}]_{\bm{i}\bm{j}}
= \begin{cases}
    (-1)^{(\bm{i}+\bm{j})_z}[\cM]_{\bm{i}\bm{j}}, & P_{\bm{j}} \in \{I, Z\}^n, \\
    0 & P_{\bm{j}} \not\in \{I, Z\}^n.
\end{cases}
\end{align}

The POVM representation $\{ E_{\bm{x}}^{\operatorname{XY}}\}_{\bm{x}}$ of $\cM^{\operatorname{XY}}$
can be computed from its PTM matrix using Proposition~\ref{prop:from-PTM-to-POVM} as
\begin{align}\label{eq:DTrCxDJBhWJBuPCDUx1}
     E_{\bm{x}}^{\operatorname{XY}}
=  \sum_{P_{\bm{i}},P_{\bm{j}}\in\mathsf{P}^n}[\mathsf{P}]_{\bm{i}\bm{x}}[\cM^{\operatorname{XY}}]_{\bm{i}\bm{j}}P_{\bm{j}}
=  \sum_{P_{\bm{i}},P_{\bm{j}}\in\{I,Z\}^{\otimes n}}(-1)^{(\bm{i}+\bm{j})_z}[\mathsf{P}]_{\bm{i}\bm{x}}[\cM]_{\bm{i}\bm{j}}P_
{\bm{j}}.
\end{align}
For arbitrary $\bm{y}\neq\bm{z}$, it holds that
\begin{align}
    \bra{\bm{y}} E_{\bm{x}}^{\operatorname{XY}}\ket{\bm{z}}
= \sum_{P_{\bm{i}},P_{\bm{j}}\in\{I,Z\}^{\otimes n}}(-1)^{(\bm{i}+\bm{j})_z}[\mathsf{P}]_{\bm{i}\bm{x}}[\cM]_{\bm{i}\bm{j}}\bra{\bm{y}}P_{\bm{j}}\ket{\bm{z}}
= 0,
\end{align}
indicating that the POVM elements after XY twirling have zero non-diagonal elements.

\section{Proof of Proposition~\ref{prop: Pauli twirling}}\label{appx:pauli-twirling}

From Eq.~\eqref{eq:DteX3}, we know that
the PTM matrix of the Pauli-twirled measurement channel $\cM^{\operatorname{Pauli}}$ has the form
\begin{align}
    [\cM^{\operatorname{Pauli}}]_{\bm{i}\bm{j}}
=  \tr\left[P_{\bm{i}}\cM^{\operatorname{Pauli}}(P_{\bm{j}})\right]
=  \left(\frac{1}{2^n}\sum_{P_{\bm{k}}\in\{I,X,Y,Z\}^{\otimes n}}
    (-1)^{\langle\bm{i}+\bm{j},\bm{k}\rangle}\right)[\cM]_{\bm{i}\bm{j}}.
\end{align}
Following a similar argument as that in Appendix~\ref{appx:xy-twirling}, we conclude that
\begin{align}\label{eq:DTrCxDJBhWJBuPCDUx2}
    [\cM^{\operatorname{Pauli}}]_{\bm{i}\bm{j}}
= \begin{cases}
    [\cM]_{\bm{i}\bm{j}}, & P_{\bm{i}}=P_{\bm{j}} \in \{I, Z\}^n,\\
    0, &  \text{otherwise}.
\end{cases}
\end{align}
The POVM representation $\{ E_{\bm{x}}^{\operatorname{Pauli}}\}_{\bm{x}}$ of $\cM^{\operatorname{Pauli}}$
can be computed from its PTM matrix using Proposition~\ref{prop:from-PTM-to-POVM} as
\begin{align}
     E_{\bm{x}}^{\operatorname{Pauli}}
=  \sum_{P_{\bm{i}},P_{\bm{j}}\in\mathsf{P}^n}[\mathsf{P}]_{\bm{i}\bm{x}}[\cM^{\operatorname{Pauli}}]_{\bm{i}\bm{j}}P_{\bm{j}}
=  \sum_{P_{\bm{i}}=P_{\bm{j}}\in\{I,Z\}^{\otimes n}}[\mathsf{P}]_{\bm{i}\bm{x}}[\cM]_{\bm{i}\bm{j}}P_{\bm{j}}
=  \sum_{P_{\bm{i}}\in \{I,Z\}^{\otimes n}}[\mathsf{P}]_{\bm{i}\bm{x}}[\cM]_{\bm{i}\bm{i}}P_{\bm{i}}. \label{eq:povm after pauli twirling}
\end{align}
For arbitrary $\bm{y}\neq\bm{z}$, it holds that
\begin{align}
    \bra{\bm{y}} E_{\bm{x}}^{\operatorname{Pauli}}\ket{\bm{z}}
= \sum_{P_{\bm{i}}\in \{I,Z\}^{\otimes n}}[\mathsf{P}]_{\bm{i}\bm{x}}[\cM]_{\bm{i}\bm{i}}\bra{\bm{y}}P_{\bm{i}}\ket{\bm{z}}
= 0,
\end{align}
indicating that the POVM elements after Pauli twirling have zero non-diagonal elements.

\section{Proof of Proposition~\ref{prop:fidelity-preserving}}\label{appx:fidelity-preserving}

Recall that the measurement fidelity of $\cM$ (with respect to the computational basis measurement) is defined as
\begin{align}
    f(\cM) = \frac{1}{2^n}\sum_{\bm{x}\in\{0,1\}^n}\langle \bm{x}\vert E_{\bm{x}}\vert\bm{x}\rangle.
\end{align}
Using Proposition~\ref{prop:from-PTM-to-POVM}, we express $f(\cM)$ in terms of the PTM matrix $[\cM]$ as
\begin{align}\label{eq:ZqHnxQmjFalAdYAaQ1}
    f(\cM)
&=  \frac{1}{2^n}\sum_{P_{\bm{i}},P_{\bm{j}}\in\{I,Z\}^{\otimes n}} \left(\sum_{\bm{x}}
\langle\bm{x}|P_{\bm{i}}|\bm{x}\rangle\langle\bm{x}|P_{\bm{j}}|\bm{x}\rangle\right) [\cM]_{\bm{i}\bm{j}}.
\end{align}
Let's analyze the term within the bracket in depth:
\begin{align}
    \langle\bm{x}|P_{\bm{i}}|\bm{x}\rangle\langle\bm{x}|P_{\bm{j}}|\bm{x}\rangle
&=  \sum_{\bm{x}}\langle\bm{x}|Z[\bm{i}_z]|\bm{x}\rangle\langle\bm{x}|Z[\bm{j}_z]|\bm{x}\rangle \\
&=  \frac{1}{2^n}\sum_{\bm{x}}(-1)[(\bm{i}_z+\bm{j}_z)\cdot \bm{x}] \\
&= \frac{1}{2^n}\sum_{\bm{x}'} (-1)[\bm{x}'] \\
&= \begin{cases}
    0, & \bm{i}_z\neq\bm{j}_z,\\
    1, & \bm{i}_z=\bm{j}_z,
\end{cases}\label{eq:ZqHnxQmjFalAdYAaQ2}
\end{align}
where $\bm{x}' = (\bm{i}_z+\bm{j}_z)\cdot\bm{x}$.
Substituting Eq.~\eqref{eq:ZqHnxQmjFalAdYAaQ2} back to Eq.~\eqref{eq:ZqHnxQmjFalAdYAaQ1} yields
\begin{align}
    f(\cM) = \frac{1}{2^n}\sum_{P_{\bm{i}}\in\{I,Z\}^{\otimes n}}[\cM]_{\bm{i}\bm{i}},
\end{align}
indicating that $f(\cM)$ is determined only by the diagonal elements of its PTM matrix. 
As evident from Eqs.~\eqref{eq:DTrCxDJBhWJBuPCDUx0}, ~\eqref{eq:DTrCxDJBhWJBuPCDUx1} and~\eqref{eq:DTrCxDJBhWJBuPCDUx2}, 
the PTM matrices of the IZ dephased, XY twirled, and Pauli twirled measurement channels
have the same diagonal elements as that of the original measurement, 
we thus conclude that the effective measurements generated by these elimination methods
have the same measurement fidelity as that of the original measurement.

\section{Proof of Proposition~\ref{prop:regularize povm}}\label{appx:regularize povm}

Recall Eq.~\eqref{eq:povm after pauli twirling} that the PTM matrix of the Pauli-twirled measurement channel satisfies
\begin{align}
     E_{\bm{x}}^{\operatorname{Pauli}}
= \sum_{P_{\bm{i}}\in \{I,Z\}^{\otimes n}}[\mathsf{P}]_{\bm{i}\bm{x}}[\cM]_{\bm{i}\bm{i}}P_{\bm{i}}
= \sum_{s_z} (-1)[\bm{i}_z\cdot \bm{x}] [\cM_z]_{\bm{i}_z} Z[\bm{i}_z],
\end{align}
where $[\cM_z]_{\bm{i}_z} = [\cM]_{\bm{i}\bm{i}}$ and $\bm{i}_z$ is a
binary string. Note that we omit $\bm{i}_x$ since $\bm{i}_x = 0\cdots 0$ for any $P_{\bm{i}}\in\{I,Z\}^{\otimes n}$.

Let's analyze the diagonal elements of $ E_{\bm{x}}^{\operatorname{Pauli}}$.
For arbitrary $\bm{y}$, we have
\begin{align}
  \langle \bm{y} |  E_{\bm{x}}^{\operatorname{Pauli}} | {\bm{y}} \rangle
= \sum_{\bm{i}_z} (-1)[(\bm{x}+\bm{y})\cdot \bm{i}_z][\cM_z]_{\bm{i}_z}.
\end{align}
We can express the above equations (for all $\bm{y}$) in the matrix form
\begin{align}\label{eq:etNN}
    \operatorname{diag}( E_{\bm{x}}^{\operatorname{Pauli}}) = \mathsf{T}_{\bm{x}} \operatorname{diag}(\cM),
\end{align}
where $\operatorname{diag}(A)$ represents the diagonal column vector of a operator $A$ and
\begin{align}
    \mathsf{T}_{\bm{x}} :=
    \begin{bmatrix}
        1 & \cdots & (-1)[{\bm{x}}] \\
        \vdots & \ddots & \vdots \\
        1 & \cdots & (-1)[(d-1)+{\bm{x}}]
    \end{bmatrix},\qquad
    \left(\mathsf{T}_{{\bm{x}}}\right)_{\bm{i}\bm{j}} := (-1)[({\bm{x}}+\bm{i})\cdot \bm{j}].
\end{align}
One can show that rank of $\mathsf{T}_{\bm{x}}$ equals to $2^n$ and hence it is invertible.
From Eq.~\eqref{eq:etNN} we obtain the following relation
between arbitrary two POVM elements $ E_{\bm{x}}^{\operatorname{Pauli}}$
and $ E_{\bm{y}}^{\operatorname{Pauli}}$:
\begin{align}
    \operatorname{diag}( E_{\bm{y}}^{\operatorname{Pauli}})
= \mathsf{T}_{\bm{y}} \mathsf{T}_{\bm{x}}^{-1} \operatorname{diag}( E_{\bm{x}}^{\operatorname{Pauli}}),
\end{align}
showing that they share the same diagonal values and
the order of the values is completely specified by the indices $\bm{x}$ and $\bm{y}$.

\section{VQE configuration}\label{appx:vqe-configuration}

\subsection{Hamiltonian of a hydrogen molecule}

We obtain from the OpenFermion library the following $4$-qubit Hamiltonian describing a hydrogen molecule,
expressed in terms of Pauli strings:
\begin{align}\label{eq:hydrogen molecule}
    H_2 &= -\; 0.097066 I - 0.045303 X_0X_1Y_2Y_3 + 0.045303 X_0Y_1Y_2X_3 + 0.045303 Y_0X_1X_2Y_3\nonumber \\
&\quad - 0.045303 Y_0Y_1X_2X_3 + 0.171413 Z_0 + 0.168689 Z_0Z_1 + 0.120625 Z_0Z_2\nonumber \\
&\quad + 0.165928 Z_0Z_3 + 0.171413 Z_1 + 0.165928 Z_1Z_2 + 0.120625 Z_1Z_3\nonumber \\
&\quad - 0.223432 Z_2 + 0.174413 Z_2Z_3 - 0.223432 Z_3.
\end{align}
In the above notation, $X_i$ means Pauli operator in the $i$-th qubit
and the identity matrix is omitted for simplicity.
For example, the term $Z_0$ should be understood as $Z_0\otimes I_1\otimes I_2\otimes I_3$.

\subsection{Transformed Hamiltonian}

Here we show that in a VQE algorithm, if the measurement is a Ry measurement as 
conceived in Section~\ref{sec:detection-simulation}, VQE would estimate 
the ground state energy of a new Hamiltonian $\wt{H}$ transformed from the target Hamiltonian $H$.
The argument goes as follows.

In VQE, the target Hamiltonian $H$ is decomposed into a linear combination of Pauli operators
$H=\sum_i\alpha_iP_i$. We first estimate the expectation values of these Pauli operators and then combine these values.
In order to estimate the expectation value $\tr[P_i\rho]$ of a given Pauli operator $P_i$,
when $\rho$ is the quantum state produced by an ansatz,
we need to add a basis transform circuit to rotate the Pauli measurement
to the standard $Z$ measurement. More precisely, assume the 
spectral decomposition $P_i=\beta_{j\vert i}\phi_{j\vert i}$, where
$\{\vert\phi_{j\vert i}\rangle\}_j$ forms an orthonormal basis.
Let $Q_i$ be a unitary that transforms the basis $\{\vert\phi_{j\vert i}\rangle\}_j$ to the computational basis,
i.e., $\forall j, Q_i\vert\phi_{j\vert i}\rangle = \ket{j}$. We have
\begin{align}\label{eq:VQE-expectation}
  \tr[P_i \rho]
= \tr\left[\left(\sum_j\beta_{j\vert i}\phi_{j\vert i}\right)\rho\right]
= \tr\left[\left(\sum_j\beta_{j\vert i}Q_i^\dagger\proj{j}Q_i\right)\rho\right]
= \sum_j\beta_{j\vert i}\bra{j}Q_i\rho Q_i^\dagger\ket{j}.
\end{align}
Experimentally, we can add quantum gates that implement $Q_i^\dagger$ to the quantum state $\rho$ and
then execute the $Z$ measurement to estimate the expectation value.

However, in our noisy simulation, we add $R_y$ gates before the standard $Z$ measurement to mimic a noisy measurement.
In this case, Eq.~\eqref{eq:VQE-expectation} becomes
\begin{align}
  \tr[P_i \rho]
= \tr\left[\left(\sum_j\beta_{j\vert i}Q_i^\dagger{\color{red}R_y^\dagger\proj{j} R_y} Q_i\right)\rho\right]
= \tr\left[Q_i^\dagger R_y^\dagger\left(\sum_j\beta_{j\vert i}\proj{j}\right)R_y Q_i\rho\right].
\end{align}
That is, each Pauli operator $P_i$ is effectively transformed as follows
\begin{align}
    P_i \quad\mapsto\quad Q_i^\dagger R_y^\dagger\left(\sum_j\beta_{j\vert i}\proj{j}\right)R_y Q_i.
\end{align}
Correspondingly, the transformed Hamiltonian has the form
\begin{align}\label{eq:VQE-expectation-2}
\wh{H} := \sum_i\alpha_iQ_i^\dagger R_y^\dagger\left(\sum_j\beta_{j\vert i}\proj{j}\right)R_y Q_i.
\end{align}

For the target hydrogen molecule Hamiltonian $H_2$~\eqref{eq:hydrogen molecule},
the transformed Hamiltonian $\wh{H}_2$ can be analytically computed using~\eqref{eq:VQE-expectation-2}
and has a ground state energy $-1.135$.

\end{appendices}

\end{document}